\documentclass[lettersize,journal]{IEEEtran}
\usepackage[T1]{fontenc}
\usepackage[a4paper]{geometry}
\usepackage{amsmath}
\usepackage{color}
\usepackage{graphicx}
\usepackage{amssymb}
\usepackage{amsthm}
\usepackage[normalem]{ulem}
\makeatletter

\newtheorem{theorem}{Theorem}
\newtheorem{lemma}{Lemma}



\def\cA{{\cal A}}
\def\cB{{\cal B}}
\def\cC{{\cal C}}
\def\cD{{\cal D}}

\def\cU{{\cal U}}

\def\cI{{\cal I}}

\def\mR{ \mathbb{R}}

\def\mC{{\mathbb{C}}}
\def\cA{{\mathcal{A}}}
\def\cB{{\mathcal{B}}}
\def\cC{{\mathcal{C}}}

\newtheorem{assumption}{Assumption}

\makeatother

\usepackage{babel}

\ifCLASSINFOpdf
\else
\fi
\usepackage[caption=false,font=footnotesize]{subfig}

\usepackage{url}

\begin{document}
	
	\title{Reverse-Time Diffusion Processes for Discrete  Time Linear and Nonlinear Systems with Non-Gaussian Noise}
	\author{Soura Dasgupta, \IEEEmembership{Life Fellow, IEEE}, Brian D. O. Anderson, \IEEEmembership{Life Fellow, IEEE},
		and Raghuraman Mudumbai, \IEEEmembership{Member, IEEE}
		\thanks{Supported in part by   NSF Grant  2540120. Dasgupta and Mudumbai are with the Department of Electrical \& Computer Engineering, University of Iowa, Iowa City, IA 52242. Anderson is with the School of Engineering, The Australian National University, ACT 2601, Australia.  email:\{dasgupta,rmudumbai\}@engineering.uiowa.edu, brian.anderson@anu.edu.au.} 
	}
\maketitle

\begin{abstract}
 Generative AI  relies on finding  reverse time models for a  discrete-time forward diffusion with non-Gaussian initial state, but uses indirect approaches  as there is no  theory for direct reversal in discrete time. This paper develops a theory for directly finding  reverse diffusions for  discrete time nonlinear processes with non-Gaussian states and  process noise. We also give a necessary and sufficient condition for the reverse model to be input-affine when the forward process is linear and the process noise Gaussian, and show that for a wide variety of state densities an input-affine reverse diffusion does not exist. This is among several differences between the reversal of stochastic difference equations and their continuous time counterparts.
\end{abstract}


%
\IEEEpeerreviewmaketitle

\section{Introduction}\label{sintro}

Reversing a forward stochastic discrete-time diffusion is fast becoming a core component of diffusion based Generative AI, \cite{Sazara}. Examples include image generation and enhancement \cite{NEURIPS2020_DDPM}, \cite{song2019generative}, \cite{song2021scorebased}, \cite{DhariwalNEURIPS2021_49ad23d1}, \cite{CroitoruSurvey},   audio synthesis, \cite{kong2021diffwave}, and natural language processing, \cite{LiText2020}. There is compelling evidence that this technology outperforms  generative adversarial networks (GANs), \cite{DhariwalNEURIPS2021_49ad23d1}.

The main idea  is as follows. The forward diffusion progressively adds noise until for example an image becomes completely blurred, typically through a linear time-varying (LTV) discrete-time process like
\begin{equation}\label{linearnoising}
	x_{k+1}=\sqrt{1-\epsilon_{k}}x_k+\sqrt{\epsilon_{k}}u_k
\end{equation}
where $u_k\sim N(0,I)$ is white and independent of the current and past states $x_i$ and $0<\epsilon_{k}\leq 1$ is progressively increased from a small  number to one, \cite{NEURIPS2020_DDPM}. The
reverse diffusion then performs a denoising operation, \cite{song2021scorebased}.

We are motivated by the way in which this reverse diffusion  is obtained. Instead of a direct reversal, one first obtains a stochastic  differential equation (SDE)  by taking the limit (as the associated discretization interval tends to zero) of  the forward difference equation. A reverse continuous-time diffusion is then obtained using theory developed  in \cite{anderson1982reverse}. Finally, the reverse difference equation is obtained by discretizing the reverse SDE. 

This approach of using an approximation of another approximation is fraught with issues. First, approximations always induce errors. Second, discretizing a reverse SDE to sufficient fidelity  is computationally onerous. A state space of as modest a dimension as 256 may take as many as 48 hours to accomplish reversal using standard packages like \cite{Package}. More fundamentally, \cite{clements} presents SDEs with bounded solutions whose discretizations always lead to  unbounded solutions.

\textit{Evidently, this indirect approach is compelled by the absence of a theory for directly reversing a stochastic difference equation.} The only exception, barring \cite{cdc25} (the preliminary version of this paper),  is \cite{Kailath} (see \cite{Vergehesecorrection} for a correction) which  reverses a  discrete-time LTV system with Gaussian input and \textit{initial Gaussian state}. Applications of Generative AI generally prohibit  a Gaussian initial state. Accordingly, in this paper we provide a theory for directly obtaining a discrete-time reverse process for a forward nonlinear time-varying (NLTV) diffusion in which  with $x_k\in \mR^n$, a \textit{potentially non-Gaussian sequence} $u_k\in \mR^m$ and  a continuously differentiable $f_k(\cdot,\cdot):\mR^{n\times m}\rightarrow \mR^n$, there holds
\begin{equation}\label{nonforward}
	x_{k+1}=f_k(x_k,u_k), ~\forall~k_0\leq k\leq k_f.
\end{equation}
We further assume that $u_k$, is white, i.e. its samples  are jointly independent, and  the $u_k$ and $x_{k_0}$ sequence  are jointly independent. This is in keeping with the notion of a forward process and ensures $u_k$ and $x_i$ are independent for all $i\leq k$. This model thus significantly generalizes and goes well beyond \cite{cdc25} that gave preliminary results without proofs for reversing an LTV forward diffusion
\begin{equation}\label{vecforward}
	x_{k+1}=F_k x_k+G_k u_k,, ~\forall~k_0\leq k\leq k_f,
\end{equation}
where at each time $k$, $F_k\in\mathbb R^{n\times n},~ G_k\in\mathbb R^{n\times m}$ with $u_k\sim N(0,I)$. \textit{Crucially, in keeping with the exigencies of Generative AI, unlike \cite{Kailath}, neither \cite{cdc25} nor this paper  requires  the initial state to be Gaussian.}  This makes the problem nontrivial, as successive states may have different types of  densities, and even a time-invariant version of \eqref{vecforward} will induce a nonlinear time-varying reverse diffusion. Separately, we note that even for \eqref{vecforward}, the results in \cite{cdc25} \textit{are  a subset of those  here}.

In  \cite{anderson1982reverse}, a method is given for the construction from
\begin{equation}\label{eq:forward}
	dx=f(x)dt+g(x)du
\end{equation}
where  $x\in\mathbb R^n$ and $du$ is an increment of an $m$-dimensional Wiener process, a \textit{reverse-time} diffusion equation of the form
\begin{equation}\label{eq:backward}
	dx=f_r(x)dt+g_r(x)dv.
\end{equation}
The distinction between the forward-time equation and the reverse-time equation is that for the former, an increment $du(t)$ is independent of $x(s)$ for all $t>s$, so that the equation can  be seen as evolving forward in time, whereas for the reverse-time equation, an increment $dv(t)$ is independent of $x(s)$ for all $t<s$, so that the equation can naturally be thought of as evolving backward in time. Both equations produce the same sample functions.

This paper, develops a theory for directly reversing (\ref{nonforward}) i.e. obtaining a reverse discrete-time process:
\begin{equation}\label{nonbackward}
	x_{k}=a_{k+1}(x_{k+1},v_{k+1}), ~\forall~k_0\leq k\leq k_f
\end{equation} 
together with a second equation constituting a formula for the reverse time noise 
\begin{equation}\label{genv}
	v_k=h_k(x_k,u_{k-1}), ~\forall~k_0< k\leq k_f+1,
\end{equation}
with $a_k(\cdot,\cdot):\mR^{n\times m}\rightarrow \mR^n$ and  
$h_k(\cdot,\cdot):\mR^{n\times m}\rightarrow \mR^m$ continuously differentiable.
The sample paths in (\ref{nonbackward}) must match those in (\ref{vecforward}). To justify designating $v_k$ as the noise in a reverse-time diffusion process, we require that $v_{k}$ and $x_i$ be jointly independent for all $k\leq i$ and  $v_{k}$  be white, though \textit{not necessarily Gaussian}. 

Stochastic difference equations do not always inherit  properties of SDEs.  The example of \cite{clements} provides an SDE whose solution is always bounded, but for which any standard time-discretization is unbounded. This is because there exists for this SDE no stepsize applicable to  the \textit{semi-infinite} interval $[t_0,\infty)$ that guarantees stability of an Euler approximation. 
Unsurprisingly, many properties of reverse SDEs do not hold in discrete time.   

 One may ask if the conditional probability $p_{X_k|X_{k+1}}(x_k|x_{k+1})$ in the reverse model for \eqref{vecforward} is Gaussian is when $u_k\in N(0,I)$. This  likely requires that for some $A_{k+1}(\cdot): \mR^n\rightarrow \mR^n$, $B_{k+1}(\cdot): \mR^n\rightarrow \mR^{n\times m}$, 
\begin{equation}\label{vectreverse}
	x_k=A_{k+1}(x_{k+1})+B_{k+1}(x_{k+1})v_{k+1}
\end{equation}
with $v_{k+1}$ Gaussian. Indeed   (\ref{eq:backward}) is \textit{input-affine} (in $dv$) and $dv$ is also the increment of a Wiener process. Thus, we explore if \eqref{vectreverse} can be the reverse model  for \eqref{vecforward}. While we show that if in \eqref{vecforward} the initial state and $u_k$ are Gaussian, the reverse model is indeed as in (\ref{vectreverse}), and in fact linear with Gaussian $v_{k+1}$, we prove that for a wide class of initial state distributions the reverse model for \eqref{vecforward} cannot be input affine   even when $u_k\sim N(0,I)$. This includes  states with  
Gaussian mixture densities, which are of great interest as they   approximate most densities, \cite{parzen1962estimation}, 
\cite{li1999mixture}, \cite{van2024mixtures}.

Section \ref{sprel} defines the problem. Section \ref{sind} shows that the independence of $v_k$ and $x_k$ ensures that $v_k$ is white and independent of all $x_{k+i}$, $i\geq 0.$ It also gives a  necessary and sufficient condition for this independence. Section \ref{sbij} shows how to ensure that  forward and reverse models generate the same sample paths. Section \ref{sreverse} shows that a reverse diffusion exists if in (\ref{genv}) every $x_{k+1}$, $v_{k+1}$ pair uniquely determines $u_k$, and also $v_k$ and $x_k$ are independent. It shows that the reverse model for \eqref{vecforward} will be input-affine only if \eqref{genv} is affine in $v_{k+1}$.  Section \ref{ssuff} provides  methods for meeting these conditions guaranteeing the existence of a reverse model (though it need not be input-affine), and gives a family of $h_{k}(\cdot,\cdot)$  that ensure that $v_{k}$ and $x_{k}$ are independent. Section \ref{sneg}, deals with \eqref{vecforward} and is the first section that assumes a Gaussian $u_k$. It shows that our results recover those of \cite{Kailath} and also  that a wide class of forward LTV processes do not have input-affine reversals. Section \ref{sex} gives the important example when the state densities are Gaussian Mixtures. Section \ref{sconc}, the conclusion,  includes a discussion on how the ideas might be made  directly applicable to Generative AI.

\section{Problem Statement and Preliminaries}\label{sprel}

We make the following assumptions.
\begin{assumption}\label{ass:two}
	\begin{enumerate}
		\item The  function $f_k:\mR^{n\times m}\rightarrow \mR^n$ is measurable and continuously differentiable.
        \item
		Every pair $x_{k+1}$ and $u_k$ uniquely determines $x_k$. More precisely, there  is a  continuously differentiable implicit function  $g_k:\mR^{n\times m}\rightarrow \mR^n$ such that
		\begin{equation}\label{implicitforward}
			x_{k}=g_{k+1}(x_{k+1}, u_k).
		\end{equation} 
		
		For \eqref{vecforward} this means $F_k$ is nonsingular.
		\item 
		Every pair $x_{k+1}$ and $x_k$ uniquely determines $u_k$ through a continuously differentiable implicit function. For \eqref{vecforward} this means rank$[G_k]=m$.
		\item 
		The probability density function (pdf) $p_{X_k}(x)$ of $x_{k}$ exists as does $p_{U_k}(u)$ of $u_{k}$. Moreover, the joint pdf $p_{X_{k+1}U_k}(x,u)\neq 0$ for all $x\in\mR^n$ and  $u\in\mR^m$. 
        \item The state is independent with current and future process noise samples, or formally, $x_{k_0}, u_{k_0}, u_{k_0+1}, \cdots, u_{k_f+1} $ are jointly independent.
	\end{enumerate}
\end{assumption}
Note \textit{we do not need either the initial state nor $u_k$ to be  Gaussian or zero mean.} The measurability of $f_k(\cdot,\cdot)$ is a standard technical condition.
The second and third requirements enable the sample paths in the reverse and forward models to match. In fact for \eqref{vecforward} their violation  precludes the possibility of recovering certain sample paths \cite{Kailath}. The third requirement also precludes degenerate situations where  the state is unaffected by the input process.  The condition on the joint pdfs also holds in most reasonable settings. One of its consequences is that $p_{X_{k}}(x) \neq 0 $ and  $p_{U_{k}}(u) \neq 0 $ everywhere. It is noteworthy that Generative AI algorithms work with the score function $\nabla \log p_{X_k}(x)$ which is undefined if $p_{X_{k}}(x) = 0 $.

Finally, 5) has been motivated in the introduction and is fundamental to the notion of a forward diffusion. It also implicitly affirms the whiteness of the process noise sequence $u_k$. Strictly speaking, the sample $u_{k_f+1}$ is not needed either in the forward or the reverse diffusion, but its inclusion simplifies the  narrative to follow.

We now define the  problem. 

\noindent
\textbf{Problem Statement:} Consider the forward diffusion  (\ref{nonforward}) with  Assumption \ref{ass:two} in force. Find a reverse time diffusion \eqref{nonbackward} and a backward noise sequence defined by (\ref{genv})   so that:
1) $x_k$  generated by \eqref{nonforward} obeys \eqref{nonbackward} under (\ref{genv}) and
	 the mapping between $(x_k,u_k)$ and $(x_{k+1},v_{k+1})$ is a bijection;
	2) 
	$v_k$ and $x_i$ are independent for all $i\geq k$;
	3)
	 $v_{k_0+1}, \cdots, v_{k_f+1}$ are jointly independent.
     
     \vspace{1 mm}
     
The second  property is motivated in the Introduction and the last captures the requirement that $v_k$  be white.  \textit{We do not require either $v_k$ or $u_k$ to be Gaussian.}  Due to the Markovian nature of the forward system \eqref{nonforward} (or the specialization (\ref{vecforward})), and separately of the reverse-time system (\ref{nonbackward}) and (\ref{genv}), requirement 1) ensures the following. Suppose we choose $x_{k_f+1}$ in   (\ref{nonbackward}) and (\ref{genv}) to be as generated by \eqref{nonforward} or (\ref{vecforward}), and choose $v_{k_f+1}=h_{k_f+1}(x_{k_f+1},u_{k_f}) $. Then the state sequences generated by \eqref{nonforward} or (\ref{vecforward}) and    (\ref{nonbackward}) under (\ref{genv}) are the same sequence of random vectors, making (\ref{nonbackward}) under (\ref{genv}) a  reverse diffusion for the forward process for $k_0\leq k\leq k_f$. 
\textit{We  say that \eqref{nonbackward} is a  reverse diffusion  if 1)-3) in the problem statement are met.}

\section{Whiteness of $v_k$ and  independence with future states}\label{sind}
We now show that under \eqref{nonforward} and \eqref{genv}, the requirements that  $v_k$ and $x_i$ be independent for all $i\geq k$ and  $v_k$ be  white, are  both satisfied if  the random variables $v_{k}$ and $x_k$ are independent. The existence and construction of \eqref{genv} meeting this requirement is given later. 

\begin{lemma} \label{lwhite}
	Under Assumption \ref{ass:two} suppose \eqref{nonforward} holds and for some measurable $h_k(\cdot,\cdot):\mR^{n\times m}\rightarrow \mR^m$,
	(\ref{genv}) holds.
	Assume the following apply: (i)  $x_k, u_k, u_{k+1}, \cdots, u_{k_f+1}$  are jointly independent  for all $k_0\leq k\leq k_f+1$; and (ii) for all $k_0\leq k\leq k_f$, $v_{k+1}$ and  $x_{k+1}$ are  independent. Then the following hold: (A) 	 $v_{k+1}$ and $x_{k+i+1}$,   are   independent  for all $0\leq i\leq k_f-k$ and $k_0\leq k\leq k_f$
	 and (B) $v_{k_0+1},\cdots,v_{k_f+1}$ are jointly independent.
\end{lemma}
\begin{proof}
	Because  \eqref{genv} is a measurable function of $x_k$ and $u_{k-1}$, from (i) of the lemma hypothesis, $(v_k,x_k)$ and $(u_k,  \cdots, u_{k_f+1})$ are  independent. For Borel sets $\cA,\cB, \cC$
	\begin{align*}
		&\mbox{Pr}\left (v_k\in \cA, x_k\in \cB, (u_k,\cdots, u_{k_f+1})\in \cC\right )\notag\\
		&=\mbox{Pr}\left (v_k\in \cA, x_k\in \cB\right ) \mbox{Pr}\left ((u_k,\cdots, u_{k_f+1})\in \cC\right )\notag\\
		&=\mbox{Pr}\left (v_k\in \cA\right ) \mbox{Pr}\left (x_k\in \cB\right ) \mbox{Pr}\left ((u_k,\cdots, u_{k_f+1})\in \cC\right )\notag\\
		&=\mbox{Pr}\left (v_k\in \cA\right )  \mbox{Pr}\left (x_k\in \cB, (u_k,\cdots, u_{k_f+1})\in \cC\right ).
	\end{align*}
Here the penultimate equation comes from (ii) and the last from (i) of the lemma hypothesis. Thus
\begin{equation}\label{vxind}
v_k \mbox{ and }(x_k, u_k,\cdots, u_{k_f+1}) \mbox{ are independent.}
\end{equation}
The repeated application of \eqref{nonforward} yields  measurable functions $\bar{f}_{ki}$ such that for all $i\geq 1,$
\begin{equation}\label{fki}
	x_{k+i}=\bar{f}_{ki}(x_k, u_k, \cdots, u_{k+i-1}).
\end{equation}
Consequently, $(x_k,x_{k+1},\cdots,x_{k_f+1})$ is a measurable function of the sigma algebra
$\sigma(x_k, u_k, \cdots, u_{kf+1}).$
 Thus from  \eqref{vxind}, $v_k$ and $(x_k, x_{k+1},\cdots, x_{k_f+1})$ are independent.
 This proves (A). 

Consider any $j>k$.  From \eqref{fki} 
  $v_j=h_j(x_j,u_{j-1})$ is measurable with respect to
  $\sigma(x_k, u_k, \cdots, u_{j-1}).$
 Thus from \eqref{vxind}, $v_k$ and $(v_{k+1}, \cdots, v_{j}) $ are independent.  For $k_1< \cdots <k_m $ and Borel sets $\cA_1,\cdots, \cA_m$, Pr$(v_{k_1}\in \cA_1, v_{k_2}\in \cA_2, \cdots, v_{k_m}\in \cA_m)$ equals Pr$(v_{k_1}\in \cA_1)$ Pr$( v_{k_2}\in \cA_2, \cdots, v_{k_m}\in \cA_m )$.
 Repeated factorization leads to
 Pr$(v_{k_1}\in \cA_1, v_{k_2}\in \cA_2, \cdots, v_{k_m}\in \cA_m)$  equaling  $\prod_{i=1}^{m}\mbox{Pr}(v_{k_i}\in \cA_i)$,
 proving (B).
 
\end{proof}

So 2) and 3) of the problem statement are met if $v_{k+1}$ and $x_{k+1}$ are independent. We now give a nonintuitive necessary and sufficient condition (which we unpack later) for $v_{k+1}$ and $x_{k+1}$ to be independent. It involves a conditional characteristic function of $h$.  

\begin{lemma}\label{lgvecind}
	Consider \eqref{genv} with the dimensions given in Lemma \ref{lwhite}.  Associate $x_{k+1}$ with $X$, $v_{k+1}$ with $V$, $u_k$ with $U$ and $h_{k+1}(\cdot,\cdot)$ with $h(\cdot,\cdot)$. Use the common notation of using small letters such as $x$ as a value of a corresponding capitalized random variable $X$. Assume $p_X(x)$ is the pdf of $X$.
	Then  $v_{k+1}$ and $x_{k+1}$ are independent  iff for some  $\Phi(\omega)$  functionally independent of $x$, 
	\begin{equation}\label{gcondchar}
		E\left[\left . e^{j\omega^{\top}h(x,u)}\right |x\right]=\Phi(\omega)
	\end{equation}
	for all $\omega\in \mR^m$ and $x$ for which $p_X(x)\neq 0$.
\end{lemma}
\begin{proof}
	The joint characteristic function \cite{papoulis1965probability} of $v_{k+1}$ and $x_{k+1}$ can be rewritten as follows:
	\begin{align}
		\Phi_{VX}(\omega_1,\omega_2)&=E\left[e^{j\omega_1^{\top}v+j\omega_2^{\top}x}\right]\notag\\
		&=E\left[e^{j\omega_2^{\top}x}e^{j\omega_1^{\top}h(x,u)}\right]\;\;\mbox{(by \eqref{genv})}\notag\\
		&=E\left[e^{j\omega_2^{\top}x}E\left[\left .e^{j\omega_1^{\top}h(x,u)}\right|x\right]\right]\notag\\
		&=E\left[e^{j\omega_2^{\top}x}\right]\Phi(\omega_1)\;\;\mbox{(by \eqref{gcondchar})}\notag\\
		&=\Phi(\omega_1)\Phi_X(\omega_2)\notag
	\end{align}
	It is standard that $V$ and $X$ are independent iff
	\[ 
	\Phi_{VX}(\omega_1,\omega_2)=\Phi_{V}(\omega_1)\Phi_{X}(\omega_2)
	\]
	which is equivalent to $p_{VX}(v_{k+1},x_{k+1})$$=p_V(v_{k+1})p_X(x_{k+1})$.  However, note that we have simply shown to this point that $p_{VX}(v_{k+1},x_{k+1})=q(v_{k+1}) p_X(x_{k+1})$ for some function $q$, the inverse Fourier transform of $\Phi$, which has not been identified as the probability density of $v_{k+1}$ nor actually even identified as a probability density. 
	
Nevertheless, as $p_{VX}(v_{k+1},x_{k+1})=p_{V\vert X}(v_{k+1}\vert x_{k+1})$ $p_X(x_{k+1})$, for those values of $x_{k+1}$ for which $p_X(x_{k+1})$ is nonzero, it is immediate that $q(v_{k+1})= p_{V|X}(v_{k+1}|x_{k+1})$, so that the conditional probability is independent of $x_{k+1}$ on the set in question.  The values of $x_{k+1}$ when this might not hold are
	by definition a set of measure zero, the probability density being zero on the set.  Therefore the conditional probability is independent of $x_{k+1}$ almost everywhere, i.e. independence holds. 
	Hence we can conclude the independence of $V$ and $X$.
	
	For the converse, 
    we shall use the quantity 
	\begin{equation}\label{GQdef}
		\Psi(\omega,x)=E\left[\left .e^{j\omega^{\top}h(x,u)}\right|x\right].
	\end{equation}

	Suppose $V$ and $X$ are independent. Then their joint characteristic function must obey
	
	\begin{align*}
		&\Phi_{V}(\omega_1)\Phi_{X}(\omega_2)=\Phi_{VX}(\omega_1,\omega_2)\\
		&=E\left[e^{j\omega_2^{\top}x}E\left[\left .e^{j\omega_1^{\top}h(x,u)}\right|x\right]\right]
		=E\left[e^{j\omega_2^{\top}x}\Psi(\omega_1,x)\right].
	\end{align*}
	Thus, $\Phi_{V}(\omega_1)E[e^{j\omega_2^{\top}x}]=E\left[e^{j\omega_2^{\top}x}\Psi(\omega_1,x)\right]$
	 i.e.,
	\[ 
	\int_{\mR^n} (\Phi_V(\omega_1)-\Psi(\omega_1,x))e^{j\omega_2^{\top} x}p_X(x)dx=0
	\]
	for all $\omega_{i}$.
	In other words, the multidimensional Fourier Transform of $(\Phi_V(\omega_1)-\Psi(\omega_1,x))p_X(x) $, with $\omega_{2}$ as the Fourier variable, is zero for all $\omega_{2}\in \mR^n.$ Thus, $(\Phi_V(\omega_1)-\Psi(\omega_1,x))p_X(x) =0$ for all $\omega_1$ and $x$ and  
	\[ 
	\Psi(\omega,x)=\Phi_V(\omega)
	\]
	for all  $x$ where $p_X(x)\neq 0$ and   $\omega$. Thus  from (\ref{GQdef}),
	(\ref{gcondchar}) holds with $\Phi(\omega)=\Phi_V(\omega)$ whenever $p_X(x)\neq 0$. 
\end{proof}

Thus, we have given  a necessary and sufficient condition for 2) and 3) of the problem statement to be met, namely that (\ref{gcondchar}) holds.  While it does not yet deal with  the existence and construction of a function $h$, it moves us a step closer to that goal. The next section  explores bijection.

\section{Satisfying bijection}\label{sbij}
We give two results. The first is a sufficient condition for the variables in the general reverse model with its defining equations (\ref{nonbackward}) and \eqref{genv} to be bijectively relate to those in the forward diffusion  \eqref{nonforward}. Second, as the existence of an input-affine reverse model like (\ref{vectreverse}) for linear forward processes is also of interest, we give a more specialized condition for such a reverse-time model to exist for a linear forward model \eqref{vecforward}. This section does not consider parts 2) and 3) of the problem statement, but the next section will draw this section together with them. 

\begin{lemma}\label{lsuff}
	 Under  Assumption \ref{ass:two}, consider  \eqref{nonforward}.
	Suppose
	there is a  continuously differentiable, $h_k(\cdot,\cdot):\mR^{n\times m}\rightarrow \mR^m$ in \eqref{genv} such that    there exists a continuously differentiable $\psi_{k+1}(\cdot, \cdot):\mR^{n\times m}\rightarrow \mR^m$ for which
	\begin{equation}\label{psi}
		u_k=\psi_{k+1}(x_{k+1}, v_{k+1}).
	\end{equation} 
	Then, with  $g_k(\cdot,\cdot):\mR^{n\times m}\rightarrow \mR^n$ the continuously differentiable implicit function given in \eqref{implicitforward},
	\begin{equation}\label{axu}
		x_k=	g_{k+1}(x_{k+1},\psi_{k+1}(x_{k+1}, v_{k+1})).
	\end{equation}
	Moreover, under (\ref{nonforward}), (\ref{axu}) and (\ref{genv}) the mapping between $(x_k,u_k)$ and $(x_{k+1},v_{k+1})$ is a bijection.
\end{lemma}
\begin{proof}
	Clearly, \eqref{axu} follows from \eqref{implicitforward} and \eqref{psi}.
	Further,  $x_k$ and $u_k$ uniquely determine $x_{k+1}$ and $v_{k+1}$ from (\ref{nonforward}) and (\ref{genv}). On the other hand (\ref{psi}) and (\ref{axu})  uniquely determine $u_k$ and $x_k$, respectively,  from $x_{k+1}$ and $v_{k+1}$. Hence the mapping between $(x_k,u_k)$ and $(x_{k+1},v_{k+1})$ is a bijection.
\end{proof}
It should be noted that the existence of $\psi_{k+1}(\cdot,\cdot)$ implicitly  relies on 3)  of Assumption  \ref{ass:two}.

Now we give a necessary and sufficient condition for \eqref{vecforward} to have an input-affine reverse model.

\begin{lemma}\label{lind}
	Consider the special case \eqref{vecforward} of \eqref{nonforward} under Assumption \ref{ass:two}. Then the following  are equivalent:
	\begin{enumerate}
		\item      
		For a $v_{k+1}\in \mR^m$, $A_{k+1}(\cdot): \mR^n\rightarrow \mR^n$, $B_{k+1}(\cdot): \mR^n\rightarrow \mR^{n\times m}$ the input affine reverse model \eqref{vectreverse}
		holds
		with rank$[B_{k+1}(x)]=m$  for all $x$;
		\item 
		For some $v_{k+1}\in\mathbb R^m,~  C_{k+1}(\cdot):\mR^n\rightarrow \mR^{m}$ and $ D_{k+1}(\cdot):\mR^n\rightarrow \mR^{m\times m}$ that is nonsingular everywhere, the following specialization of the reverse-time noise definition \eqref{genv} holds:  
		\begin{equation}\label{vexpress}
			v_{k+1}=C_{k+1}(x_{k+1})-D_{k+1}(x_{k+1})u_k.
		\end{equation}
	\end{enumerate}
	If Condition 1 holds,  $C_{k+1}(\cdot)$ and $D_{k+1}(\cdot)$ (whose existence is guaranteed by Condition 2) are given by
	\begin{align}
		C_{k+1}(x_{k+1})&=\left(B_{k+1}^{\top}(x_{k+1})B_{k+1}(x_{k+1})\right)^{-1}\nonumber\\
		&B_{k+1}^{\top}(x_{k+1})[F_k^{-1}x_{k+1}-A_{k+1}(x_{k+1})]\label{eq:CDdef}\\
		D_{k+1}(x_{k+1})&=\left(B_{k+1}^{\top}(x_{k+1})B_{k+1}(x_{k+1})\right)^{-1}\nonumber\\
		&\times B_{k+1}^{\top}(x_{k+1})F_k^{-1}G_k. \label{Ddef}
	\end{align} 
	If Condition 2 holds,  $A_{k+1}(\cdot)$ and $B_{k+1}(\cdot)$ (whose existence is guaranteed by Condition 1) are given by
	\begin{align}
		A_{k+1}(x_{k+1})&=F_k^{-1}x_{k+1}-F_k^{-1}G_kD_{k+1}^{-1}(x_{k+1})\nonumber\\
		&\times C_{k+1}(x_{k+1})\label{eq:ABdef}\\
		B_{k+1}(x_{k+1})&= F_k^{-1}G_kD_{k+1}^{-1}(x_{k+1}).\label{Bx}
	\end{align}
	Further,   the mapping between $(x_k,u_k)$ and $(x_{k+1},v_{k+1})$ is 
	a bijection. 
	
\end{lemma}
\begin{proof}
	Recall that in \eqref{vecforward},   $F_k$ and $G_k$ have full column rank.
	Suppose Condition 2 holds, i.e. there is a $C_{k+1}(\cdot)$ and a $ D_{k+1}(\cdot)$ that is nonsingular everywhere such that (\ref{vexpress}) holds.
	Then  immediately
	\[ 
	-u_k=D_{k+1}^{-1}(x_{k+1})v_{k+1}-D_{k+1}^{-1}(x_{k+1})C_{k+1}(x_{k+1}).
	\]
	This expression in conjunction with \eqref{vecforward} then yields
	\begin{align*}
		x_k&= F_k^{-1}x_{k+1}-F_k^{-1}G_kD_{k+1}^{-1}(x_{k+1})C_{k+1}(x_{k+1})\\
		&+F_k^{-1}G_kD_{k+1}^{-1}(x_{k+1})v_{k+1}.
	\end{align*}
	With the definitions of \eqref{eq:ABdef} and (\ref{Bx}),  \eqref{vectreverse} holds. Further, $B_{k+1}$ has full column rank everywhere as $ F_k^{-1}$ is nonsingular, $G_k$ has full column rank and $D_{k+1}(x_{k+1})$ is nonsingular everywhere. Hence Condition 1 holds.

	Next, suppose that Condition 1 holds. Using \eqref{vecforward} and \eqref{vectreverse} yields
	\[ 
	A_{k+1}(x_{k+1})+B_{k+1}(x_{k+1})v_{k+1}=F_k^{-1}x_{k+1}-F_k^{-1}G_ku_k,
	\]
	or equivalently
	\begin{equation} \label{eq:Bv}
		B_{k+1}(x_{k+1})v_{k+1}=F_k^{-1}x_{k+1}-A_{k+1}(x_{k+1})-F_k^{-1}G_ku_k.
	\end{equation}
	As $B_{k+1}(x_{k+1})$ has full column rank everywhere 
	\[ 
	E_{k+1}(x_{k+1})=\left (B_{k+1}^{\top}(x_{k+1})B_{k+1}(x_{k+1})\right )^{-1}B_{k+1}^{\top}(x_{k+1})
	\]
	exists.
	Then \eqref{eq:Bv} implies
	\begin{align*}
		v_{k+1}&=  E_{k+1}(x_{k+1})(F_k^{-1}x_{k+1}-A_{k+1}(x_{k+1}))\\
		&- E_{k+1}(x_{k+1})F_k^{-1}G_ku_k.
	\end{align*}
	Then \eqref{vexpress} follows from  \eqref{eq:CDdef}  and \eqref{Ddef}.
	
	It remains to show that $D_{k+1}(x_{k+1})=E_{k+1}(x_{k+1})F_k^{-1}G_k$ is nonsingular everywhere.
	To do this, we shall compute two expressions for the partial derivative $\frac{\partial x_k}{\partial u_k}$, where we regard $x_k$ as a specialization of the function $g_k(x_{k+1},u_k)$ of $x_{k+1}$ and $u_k$. First, from \eqref{vecforward}, noting that $x_k=F_{k}^{-1}x_{k+1}-F_k^{-1}G_ku_k$, we see that 
	\[ 
	\frac{\partial x_k}{\partial u_k}=-F_{k}^{-1}G_k
	\]
	This matrix has rank $m$ since $G_k$ has that property by assumption.   Alternatively, under (\ref{vexpress}) and (\ref{vectreverse})
	\begin{align*}
		x_k&=A_{k+1}(x_{k+1})+B_{k+1}(x_{k+1})C_{k+1}(x_{k+1})\\
		&-B_{k+1}(x_{k+1})D_{k+1}(x_{k+1})u_{k},
	\end{align*}
	from which there follows the alternative expression
	\[ 
	\frac{\partial x_k}{\partial u_k}=-B_{k+1}(x_{k+1})D_{k+1}(x_{k+1})
	\] 
	The expression has rank $m$ for all $x_{k+1}$, implying $D_{k+1}(x_{k+1})\in \mR^{m\times m}$ is nonsingular for all $x_{k+1}$. Condition 1 is thus established.
	
	If $D_{k+1}(x_{k+1})\in \mR^{m\times m}$ is nonsingular for all $x_{k+1}$, then the condition in Lemma \ref{lsuff} is met and
	the mapping from $(x_k,u_k)$ to $(x_{k+1},v_{k+1})$ is  a bijection.
\end{proof}

\section{Formulating the reverse model}\label{sreverse}
We now put the results of sections \ref{sind} and \ref{sbij} together to obtain conditions for the existence of a  reverse model.
We first give a theorem that provides a sufficient condition for the existence of a  reverse-time model like (\ref{nonbackward}) for the general NLTV system \eqref{nonforward}.

\begin{theorem}\label{tsuff}
	Consider the forward diffusion equation (\ref{nonforward}) under Assumption \ref{ass:two}. Suppose there exists a measurable $h_k(\cdot,\cdot):\mR^{n\times m}\rightarrow \mR^m$ that satisfies the following two conditions. (I) (Parametrized bijection condition) There exists  a continuously differentiable  $\psi_{k+1}(\cdot, \cdot):\mR^{n\times m}\rightarrow \mR^m$ for which 
	\eqref{psi} holds. (II) (Conditional characteristic function condition) With the associations in Lemma \ref{lgvecind}, (\ref{gcondchar}) holds. Then with $g_{k+1}(\cdot,\cdot)$ defined in \eqref{implicitforward} and $ a_{k+1}(\cdot,\cdot)$ in (\ref{axu}) and \eqref{nonbackward}, the three equations  (\ref{nonforward}) (forward diffusion), (\ref{nonbackward}) (reverse-time diffusion) and (\ref{genv}) (reverse-time noise) meet the conditions specified in the problem statement.
\end{theorem}
\begin{proof}
	We need to show that with the definitions given  all of 1)-3) in the Problem statement hold. Lemma \ref{lsuff} shows that if (I) holds then with $\psi_{k+1}(\cdot,\cdot)$ defined in (\ref{psi}) and $ a_{k+1}(\cdot,\cdot)$ in (\ref{axu}), the three equations  (\ref{nonforward}), (\ref{nonbackward}) and (\ref{genv})  meet 1) of the problem statement. Further from Lemma \ref{lgvecind}, (II) implies that $v_{k+1}$ and $u_k$ are independent. Then from Lemma \ref{lwhite}, (\ref{vecforward}), (\ref{nonbackward}) and (\ref{genv}) meet 2) and 3) of the problem statement (effectively, anti-causality).
\end{proof}

Next, by way of specialization, we provide a necessary and sufficient condition for the existence of an input-affine  reverse model for the linear forward model \eqref{vecforward}. 

\begin{theorem}\label{taffine}
	Consider   (\ref{vecforward}). Under Assumption \ref{ass:two} and  dimensions  in Lemma \ref{lind}, the following are equivalent.
	\begin{itemize}
		\item[(A)] There exist  $C_{k+1}(\cdot)$, $ D_{k+1}(\cdot)$ and $v_{k+1}$ such that    an  input-affine reverse-time model \eqref{vectreverse} under an input-affine reverse-time noise formula \eqref{vexpress} meets the requirements of the Problem Statement when \eqref{nonforward} is specialized to \eqref{vecforward}.
		\item[(B)] There exist $C_{k+1}(\cdot)$ and $ D_{k+1}(\cdot)$ that is nonsingular everywhere such that under the  associations in Lemma \ref{lgvecind}, the conditional characteristic function property (\ref{gcondchar}) holds with 
        $h_{k+1}(x_{k+1},u_{k})$$=C_{k+1}(x_{k+1})-D_{k+1}(x_{k+1})u_{k}.$

	\end{itemize}
\end{theorem}
\begin{proof}
	Suppose (A) holds. From  the Problem Statement and the discussion after it, this implies: (I)
    $v_{k+1}=C_{k+1}(x_{k+1})-D_{k+1}(x_{k+1})u_k$ is white and independent with $x_{k+i}$, for $i\geq 0$; and (II) $(x_{k+1}, v_{k+1})$ and $(x_k,u_k)$ form a bijection. From Lemma \ref{lwhite}, (I) is satisfied if $v_{k+1}=C_{k+1}(x_{k+1})-D_{k+1}(x_{k+1})u_k$ and $x_{k+1}$ are independent. As $p_{X_k}(x)\neq 0$ for all $x$, from Lemma \ref{lgvecind} this in turn is equivalent to (\ref{gcondchar})  with 
        $h_{k+1}(x_{k+1},u_{k+1})$$=C_{k+1}(x_{k+1})-D_{k+1}(x_{k+1})u_{k}.$

    Assume  in turn that (II) holds and \eqref{vectreverse} is a reverse model for \eqref{vecforward}. From Lemma \ref{lind}, if \eqref{vectreverse} is a reverse model then \eqref{vexpress} holds.  To establish a contradiction,  suppose at some $x_{k+1}$, $ D_{k+1}(x_{k+1})$ is singular. Then there is a nontrivial subspace of $u_k$ such that $ D_{k+1}(x_{k+1})u_k=0$. From \eqref{vectreverse}, $x_k=A_{k+1}(x_{k+1})+B_{k+1}(x_{k+1})C_{k+1}(x_{k+1})$. For these choices of $x_{k+1}$ and $u_k$, neither \eqref{vectreverse} nor \eqref{vexpress} have any dependence on $u_k$ i.e., $x_{k+1}$ and $v_{k+1}$ cannot uniquely specify  $u_k$. Thus,  $D_{k+1}(x_{k+1})$ is nonsingular everywhere and both conditions in (B) are met i.e., (A) implies (B).
	
	When (B) holds, we will show that (I), (II)  above and \eqref{vectreverse} hold with $A_{k+1}(\cdot)$ given by \eqref{eq:ABdef} and $B_{k+1}(\cdot)$ by \eqref{Bx}. This means (A) holds, and (B) implies (A), completing the proof. As $p_{X_k}(x)\neq 0$ for all $x$, by Lemma \ref{lgvecind}, $h_{k+1}(x_{k+1},u_{k})=v_{k+1}$ given in (B) and $x_{k+1}$ are independent. From Lemma \ref{lwhite}, this means that both conditions in  (I) above are met. Also, (B) implies that \eqref{vexpress} holds.   As $D_{k+1}(\cdot)$  is nonsingular everywhere, from Lemma \ref{lind}, with $A_{k+1}(\cdot)$  and  $B_{k+1}(\cdot)$ defined by \eqref{eq:ABdef} and \eqref{Bx} the affine reverse-time equation \eqref{vectreverse} follows from (\ref{vexpress}) and the forward time equation \eqref{vecforward}. Also from Lemma \ref{lind} (II)  holds as $(x_{k+1}, v_{k+1})$ and $(x_k,u_k)$ form a bijection.

	
\end{proof}

Both  theorems state that a  reverse diffusion exists if $v_{k+1}=h_{k+1}(x_{k+1},u_k)$ meets two conditions, viz. that $v_{k+1}$ be independent of $x_{k+1}$ and every pair $x_{k+1}, 
v_{k+1}$ uniquely determine $u_k$.  Lemma \ref{lind} and Theorem \ref{taffine} further show that for the reverse model of \eqref{vecforward} to be input-affine, one also needs  $h_{k+1}(x_{k+1},v_{k+1})$ to be affine in $v_{k+1}$. Ultimately, $h_{k+1}(\cdot,\cdot)$ alone is determinative and most of the remainder of this paper is an in-depth exploration of its existence and construction.

\section{Construction of reverse diffusion} \label{ssuff}
By Theorem  \ref{tsuff}, a sufficient condition for the reverse model (\ref{nonforward}), (\ref{nonbackward}) and (\ref{genv}) to meet the requirements of the problem statement is that (i) there exists a continuously differentiable $\psi_{k+1}(\cdot, \cdot):\mR^{n\times m}\rightarrow \mR^m$ for which \eqref{psi} holds and (ii) that $x_{k+1}$ and the  $ v_{k+1}$ given by the reverse time noise formula  (\ref{genv})  are independent random variables. 
Even though these together constitute  only a sufficient condition, we  show that one can always construct a family of $h_{k+1}(x_{k+1},u_k)$ that meet these requirements using  Cumulative Distribution Functions (cdf) of certain entries of $u_k$ conditioned on $x_{k+1}$ and some other entries of $u_k$. As the forward process defines these conditional cdfs, it  always provides such  functions. As will be seen, the case of vector noise sequences  becomes significantly simpler in the scalar case.  

To keep the notation simple we will at times drop the time indices i.e, use $x_{k+1}=x$, $v_{k+1}=v$, $h_{k+1}(\cdot,\cdot)=h(\cdot,\cdot)$ and $u_k=u$.  We will associate random variables with capital letters and their values by small ones. To avoid confusion with time indices we will denote the $i$-th entry of $u_k=u\in \mR^m$ by  $u^{(i)}$. For a set $\cD$, we will define the indicator function $I_{\cD}(v)$ as,
\begin{equation}\label{indicator}
	I_{\cD}(v)=\begin{cases}
		1 & \forall ~v\in \cD\\
		0 & \forall ~v\notin \cD
	\end{cases}.
\end{equation}

Section \ref{snecsuff} provides a necessary and sufficient condition for  the required independence  
Section \ref{sdist} gives two methods for constructing $h_{k+1}(\cdot,\cdot)$ using conditional distributions. For scalar $u_k$ the second provides all possible $h_{k+1}(\cdot,\cdot)$ that meet the independence condition.

\subsection{A necessary and sufficient condition for noise and state independence}\label{snecsuff}
Define the Jacobian of $h(x,u)$ with respect to $u$ as $J_{hu}(x,u):\mR^{n\times m}\rightarrow \mR^{m\times m}$, i.e., its $kl$-th entry obeys
\begin{equation}\label{Jacob}
	\left [J_{hu}(x,u)\right ]_{kl}  ~~=\dfrac{\partial h_k(x,u)}{\partial u^{(l)}}
\end{equation}
$h_k(x,u)$ being the $k$-th entry of $h(x,u)$. 
Observe, that the existence of the continuously differentiable implicit function in $\psi_{k+1}(\cdot,\cdot)$ defined in \eqref{psi} implies that 
\begin{equation}\label{junon}
	\det\left [J_{hu}(x,u)\right ]\neq 0, ~\forall~x\in \mR^n, ~u\in \mR^m.
\end{equation}
This subsection examines when an  $h(x,u)$ obeying \eqref{junon} is independent with $x$.
We begin with a lemma that uses the fact that by 4) of Assumption \ref{ass:two}, the joint pdf of $X_{k+1}$ and $U_k$ is everywhere nonzero.

\begin{lemma}\label{lnec}
	Suppose Assumption \ref{ass:two}, holds. Consider	random vectors $X\in \mR^n$ and $U\in \mR^m$  with conditional  pdf $p_{U|X}(u|x)$.  Consider  $v=h(x,u)$ with $h(\cdot,\cdot):\mR^{n\times m}\rightarrow \mR^m$ differentiable everywhere.   Define $\Gamma(x)$ to be the range of $h(x,u)$ for a given $x$ and all $u\in \mR$. 
	Suppose also that  $X$ and $V$ are independent and \eqref{junon} is true.
	Then with $p_V(v)$ the pdf of $V$, the conditional characteristic function
	\[ 
	\int_{\Gamma(x)}e^{j\omega^{\top}v}p_V(v)dv
	\]
	and $\Gamma(x)$ are functionally independent of $x$. 
\end{lemma}
\begin{proof}
	 As $X$ and $V$ are independent, $p_{V|X}(v|x)=p_V(v)$.
	Under Assumption \ref{ass:two}, $p_X(x)$ is everywhere nonzero.
	Thus by Lemma \ref{lgvecind}, because $v=h(x,u)$ and $x$ are independent random variables,  it follows that the conditional characteristic function
	\[ 
	\int_{\mR^m} e^{j\omega^{\top} h(x,u)}p_{U|X}(u|x)du
	\]
	is functionally independent of $x$. Make the substitution $v=h(x,u)$.   Then as \eqref{junon} holds and the integration variable is $u$ alone, we can rewrite this quantity as
	\begin{align*}
		\int_{\Gamma(x)} e^{j\omega^{\top} v}\dfrac{p_{U|X}(u|x)}{\left |\det\left [J_{hu}(x,u)\right ]\right |}dv
		&=\int_{\Gamma(x)} e^{j\omega^{\top} v}p_{V|X}(v|x)dv\\
		& =\int_{\Gamma(x)} e^{j\omega^{\top} v}p_{V}(v)dv
	\end{align*}
	and Lemma \ref{lgvecind} proves  the last quantity  to be functionally independent of $x$, as required to be proved. 
	
	As $p_{XU}(x,u)\neq 0$  for all $x\in \mR^n$ and $u\in \mR^m$, so is $p_{U|X}(u|x)$, and hence \[\frac{p_{U|X}(u|x)}{|\mbox{det}[h_{xu}(x,u)]|}=p_{V|X}(v|x)=p_V(v)\neq 0\] for all $v\in \Gamma(x)$.  Further,
	\[
	\int_{\Gamma(x)} e^{j\omega^{\top} v}p_{V}(v)dv=\int_{\mR^m} e^{j\omega^{\top} v}p_{V}(v)I_{\Gamma(x)}(v)dv.
	\]
	The right hand side is the Fourier Transform of $p_{V}(v)I_{\Gamma(x)}(v)$. Since this does not vary with $x$, for every pair of distinct values  $x_1$ and $x_2$
	\begin{align*}
		&p_{V}(v)I_{\Gamma(x_1)}(v)=p_{V}(v)I_{\Gamma(x_2)}(v)~~\forall v\\ &\Leftrightarrow I_{\Gamma(x_1)}(v)=I_{\Gamma(x_2)}(v)~~\forall v,
	\end{align*}
	as $p_V(v)$ is nonzero throughout its domain. Thus, $\Gamma(x)$ does not vary with $x$.
\end{proof}
Using this Lemma, we  can now state a theorem providing a necessary and sufficient condition for the independence of $V$ and $X$.
\begin{theorem}\label{tJu}
	Consider the quantities and dimensions given in Lemma \ref{lnec} under  Assumption \ref{ass:two}. Consider a differentiable   $v=h(x,u)$ for which \eqref{junon}  holds. Then,  $v$   and $x$ are independent iff the following  conditions hold.
	
	\begin{itemize}
		\item[(i)] The domain $\Gamma(x)$, the range of $v$ for a  given $x$, is independent of $x,$ i.e. for some  $\Gamma\subset \mR^m$, $\Gamma(x)\equiv\Gamma$.
		\item[(ii)] There is a function $\xi(\cdot):\mR^m\rightarrow\mR$  such that the conditional pdf of $U|X$ obeys
		\begin{equation}\label{gentrans2}
			\dfrac{p_{U|X}(u|x)}{\left |\det\left[J_{hu}(x,u)\right]\right |}=\xi(h(x,u)).
		\end{equation}
	\item[(iii)] The function $\xi(v)I_{\Gamma}(v)$ has a  Fourier Transform.
	\end{itemize}
\end{theorem}
\begin{proof}
	 If $V$ and $X$ are independent then  Lemma \ref{lnec} shows that (i) holds.  The proof of the lemma shows that \eqref{gentrans2} holds with $\xi(\cdot)I_{\Gamma}(\cdot)=p_V(\cdot)$; this defines $\xi$ on $\Gamma(x)$ and elsewhere it can be chosen arbitrarily, e.g. as zero. establishing (ii). Since any pdf necessarily has a Fourier transform, this establishes (iii),  proving necessity.
	 
	 Now suppose (i)-(iii) hold.  By Assumption \ref{ass:two}, $p_X(x)$ is nonzero everywhere.  Make the substitution $v=h(x,u)$. Under (i) $\Gamma(x)\equiv \Gamma$  for some constant domain $\Gamma$. Then as the integration variable is $u$ alone,
	 \begin{align*}
	 	\int_{\mR^m} e^{j\omega^{\top} h(x,u)}p_{U|X}(u|x)du&=\int_{\Gamma} e^{j\omega^{\top} v}\dfrac{p_{U|X}(u|x)}{\left |\det[J_{hu}(x,u)]\right |}dv\\
	 	&=\int_{\Gamma} e^{j\omega^{\top} v}\xi(v)dv\\
	 	&=\int_{\mR^m}e^{j\omega^{\top} v}\xi(v)I_{\Gamma}(v)dv
	 \end{align*}
	 which exists as $\xi(v)I_{\Gamma}(v) $ has a Fourier Transform. Moreover, it is a function of $\omega$ alone and not of $x$. Thus $V$ and $X$ are independent from Lemma \ref{lgvecind}.
\end{proof}

\subsection{$h_k\left(x_{k+1},u_k\right)$ from conditional distributions  }\label{sdist}
 We first show that among the choices of
$\xi(\cdot)$ and the domain $\Gamma$ eligible for applying Theorem \ref{tJu} the following pair always works: $\Gamma=[0,1]^m$ and
$\xi(v)=1, ~\forall v\in \Gamma.$
(There is actually great freedom in the choice of $\xi$). This  particular choice is used in 
 the following theorem to construct  an $h(\cdot,\cdot)$   that is a \textit{conditional Rosenblatt Transform}, \cite{rosenblatt1952remarks},  satisfying our sufficient conditions. Specifically, the first entry of $h(\cdot,\cdot)$ is the cumulative distribution function {cdf} of the first entry of $u_k$ conditioned  on $x_{k+1}$. For $i>1$,  $i$-th entry of $h(\cdot,\cdot)$ is the cdf of the $i$-th entry of $u_k$ conditioned  on $x_{k+1}$ and all entries of $u_k$ with index less than $i$. Notice this makes the Jacobian in Theorem \ref{tJu} lower triangular and its determinant the product of the diagonal entries. As shown in the proof, the first diagonal entry  is the pdf of the first entry of $u_k$  conditioned  on $x_{k+1}$ while for $i>1$, the $i$-th diagonal entry is likewise  the pdf of the $i$-th entry of $u_k$ conditioned  on $x_{k+1}$ and all entries of $u_k$ with index less than $i$. This means  the determinant of the Jacobian is the conditional pdf $p_{U|X}(u|x)$.
\begin{theorem}\label{tspec}
	Consider the various quantities with their dimensions given in Theorem \ref{tJu}. Define the vectors
	$\cU_i=\begin{bmatrix}
			U^{(1)},\cdots, U^{(i)}
		\end{bmatrix}^{\top},$ $ ~~i\in \{1,\cdots,m\}$
	with $U^{(i)}$ the $i$-th entry of $U$. Call $P_{V|W}(v|w)$ the conditional cdf  of $V|W$.
	Define $h_i(x,u)$,  the $i$-th entry of $h(x,u)$, by
	\begin{align}
		&	h_1(x,u)=P_{U^{(1)}|X}(u^{(1)}|x), \nonumber\\
		& h_i(x,u)=P_{U^{(i)}|\cU^{(i-1)},X}(u^{(i)}|u^{(1)},\nonumber\cdots , u^{(i-1)} ,x)\\
		&\forall ~~i\in \{2,\cdots, m\}.\label{nest}
	\end{align}
	Then $h(x,u)$ and $x$ are independent random vectors. 
\end{theorem}
\begin{proof}
	From Assumption \ref{ass:two} $h(\cdot,\cdot)$ is differentiable.
	In (\ref{nest}), $h_i(x,u)$ has no functional dependence on $u^{(i+k)}$ for $k>0$. Thus,
	$J_{hu}(x,u)$ is lower triangular and its determinant is the product of the diagonals. Hence, from (\ref{nest})
	\begin{align*}
		\det[J_{hu}(x,u)]&=\prod_{i=1}^m \frac{\partial h_i(x,u)}{\partial u^{(i)}}\\
		&=
		\frac{\partial P_{U^{(1)}|X}(u^{(1)})|x)}{\partial u^{(1)}}\\
		\times & \prod_{i=2}^m \frac{\partial P_{U^{(i)}|\cU_{i-1},X}(u^{(i)}|u^{(1)},\cdots,u^{(i-1)},x)}{\partial u^{(i)}}\\
		&=p_{U^{(1)}|X}(u^{(1)}|x)\\
		\times &\prod_{i=2}^mp_{U^{(i)}|\cU_{i-1},X}(u^{(i)}|u^{(1)},\cdots,u^{(i-1)},x)\\
		&=p_{U|X}(u|x).
	\end{align*}
	Hence  \eqref{junon} holds from Assumption \ref{ass:two}.
	As each entry of $h(x,u)$ is a conditional cdf, the range of $h(x,u)$ in (\ref{nest})  is $[0,1]^m$ for all $x$ and thus independent of $x$. In (\ref{gentrans2}) choose $\xi(v)=1, ~\forall~ v\in [0,1]^m $. Observe $\xi(v)I_{\Gamma}(v)=I_{[0,1]^m}(v)$  has a Fourier Transform. Then all conditions of Theorem \ref{tJu} are met, 
	completing the proof.
\end{proof}
Conditional distributions in (\ref{nest}) exist everywhere because of 4) of Assumption \ref{ass:two} and are strictly monotonic in $u_k$ as $p_{U_k}(u)\neq 0$ everywhere. Thus,  $v_{k+1}=h_{k+1}(x_{k+1},u_k)$  uniquely determines $u_k$ for every $x_{k+1}$. Hence, from Theorem \ref{tsuff}, $v_{k+1}=h_{k+1}(x_{k+1},u_k) $ provides the driving white noise for a reverse diffusion.  This is also a clear point of departure from the theory for reversing an SDE, where in \eqref{eq:backward} $dv$ is an increment of a Wiener process. By contrast the $v_{k+1}$ yielded by Theorem \ref{tspec} is restricted to $[0,1]^m$ and thus \textit{cannot be Gaussian even if $u_k$ is.}

The construction above has just exhibited one particular solution to the problem statement. 
Theorem \ref{tgeneralfamily}  provides a  family of $v_{k+1}=h_{k+1}(x_{k+1},u_k)$.

\begin{theorem} \label{tgeneralfamily}
	Suppose $u=[u^{(1)}, \cdots, u^{(m)}]$ and $x\in \mR^n$ are random vectors and Assumption \ref{ass:two} holds. Consider $v=h(\cdot,\cdot)$ with $h(\cdot,\cdot)$ differentiable and $h_i(\cdot,\cdot)$ its $i$-th entry, together with differentiable functions $\pi_i:\mR\rightarrow\mR$ and  $\chi_i:\mR^n\rightarrow\mR$ that obey the following: (i) With $\cU_i$ defined in Theorem \ref{tspec}
	\[ 
	\pi_1(h_1(x,u^{(1)}))=\chi_1(x)+	P_{U^{(1)}|X}(u^{(1)}|x)
	\]
	and for $i>1$
	\begin{align*}
		&\pi_i(h_i(x,u^{(1)},\cdots,u^{(i)} ))\\
		&= \chi_i(x)+	P_{U^{(i)}|\cU_{i-1},X}(u^{(i)}|u^{(1)},\cdots,u^{(i-1)},x).
	\end{align*}
	(ii) For fixed $x$, $\Gamma$ the range of $v$ in \eqref{genv} (as $u$ varies in $\mR^m$), is independent of $x$. (iii)
	With
	 $\xi_i(v_i)=\pi_i'(v_i)$ and
	 \begin{equation}\label{prod}
	 	\xi(v)= \Pi_{i=1}^m\xi_i(v_i),
	 \end{equation}
	$\xi(v)I_{\Gamma}(v)$
	has a Fourier Transform. Then 
	 $v=h(x,u)$ and $x$ are independent.
\end{theorem}

\begin{proof}
	
	From (iii)  $\xi(\cdot)$ has a Fourier Transform.  From (i) and (iii), taking partial derivatives with respect to entries of $u$ we get
	\[ 
	\xi_1(h_1(x,u^{(1)}))\dfrac{\partial h_1(x,u^{(1)})}{\partial u^{(1)}}=p_{U^{(1)}|X}(u^{(1)}|x)
	 \]
	 and for $i>1,$
	 \begin{align*}
	 	&\xi_i(h_i(x,u^{(1)},\cdots,u^{(i)}))\dfrac{\partial h_i(x, u^{(1)},\cdots, u^{(i)})}{\partial u^{(i)}}\\
	 	&=p_{U^{(i)}|\cU_{i-1},X}(u^{(i)}|u^{(1)},\cdots,u^{(i-1)},x).
	 \end{align*}
	As the $i$-th entry  of $h(x,u)$ is
	$v_i=h_i(x,u^{(1)},\cdots,u^{(i)}),$
	the Jacobian $J_{hu}(x,u)$ is lower triangular and
	\[ 
	\det[J_{hu}(x,u)]=\prod_{i=1}^m\dfrac{\partial h_i(x, u^{(1)},\cdots,u^{(i)})}{\partial u^{(i)}}.
	\]
	There thus holds,
	\begin{align}
		p_{U|X}(u|X)&=p_{U^{(1)}|X}(u^{(1)}|x)\\
		&\times
		\prod_{i=2}^mp_{U^{(i)}|\cU_{i-1},X}(u^{(i)}|u^{(1)},\cdots,u^{(i-1)},x)\nonumber\\
		&=\left (\prod_{i=1}^m \xi_i(h_i(x,u^{(1)},\cdots,u^{(i)})) \right )\\
		&\times\left (\prod_{i=1}^m\dfrac{\partial h_i(x, u^{(1)},\cdots, u^{(i)})}{\partial u^{(i)}}\right )\nonumber\\
		&=\xi(h(x,u))\det[J_{hu}(x,u)].\label{indu}
	\end{align}
	Further,  as from Assumption \ref{ass:two}, $p_{U|X}(u|x)$ is nonzero everywhere, so is $\det[J_{hu}(x,u)]$.
	Then from (ii), \eqref{indu} and the fact that $\xi(\cdot)$ has a Fourier Transform,  $v=h(x,u)$ meets all conditions of
	Theorem \ref{tJu}. Thus, $v=h(x,u)$ and $x$ are independent random variables.
\end{proof}
This thus permits distinguished functions of the $h(\cdot,\cdot)$ in Theorem \ref{tspec} to also work. The need for conditional Rosenblatt Transform to cope with the determinant of  $J_{hu}(\cdot,\cdot)$ imposes a loss of generality preventing this from being a necessary condition. \textit{However, we now show  that when $u$ is a scalar this is in fact necessary as well.}

\begin{theorem}\label{tnecc}
	Suppose under Assumption \ref{ass:two}, $u\in \mR$,  $x\in \mR^n$ and there exists a continuously differentiable $h:\mR^{n\times 1}\rightarrow \mR$ such that with $v=h(x,u)$ the following hold: (i)  $v$ and $x$ are independent and (ii)
	\begin{equation}\label{nonzero}
	\dfrac{\partial h(x,u)}{\partial u}=h_u(u,x)\neq 0, ~\forall ~x\in \mR^n, ~u\in \mR.
	\end{equation}
	 Then: (A) There   is a differentiable  $\pi:\mR\rightarrow\mR$ and a $\chi:\mR^n\rightarrow\mR$ such that
	 \[ 
	 \pi(h(x,u))=\chi(x)+P_{U|X}(u|x).
	  \]
	  (B) The range $\Gamma$ of $h(x,u)$ for $u\in \mR$ is independent of $x$.
	  (C) The function $\pi'(v)I_{\Gamma}(v)$ has a Fourier Transform.
\end{theorem}
\begin{proof}
Theorem \ref{tJu} proves    that there is a $\xi(\cdot)$ such that $\xi(v)I_{\Gamma}(v)$ has a Fourier Transform  and  \eqref{gentrans2} and (B) hold.  Existence of the Fourier Transform ensures that  $\xi(\cdot)$ is integrable over its domain $\Gamma$. As $h_u(x,u)$ is continuous, \eqref{nonzero} implies that $h_u(x,u)$ is either always positive or negative.  Define $\pi(\cdot)$ to be the integral of $\xi(\cdot)$ if $h_u(x,u)>0$. Choose it to be negative of the integral otherwise.  Thus (C) holds and from \eqref{gentrans2},
 \[ 
 \dfrac{\partial \pi(h(x,u))}{\partial u}=\xi(h(x,u))|h_u(x,u)|=p_{U|X}(u|x).
 \]
Consequently
for some $\chi(\cdot)$, (A) holds as
\begin{align*}
	\pi(h(x,u))&= \chi(x)+\int_{\mR}p_{U|X}(u|x)du.
\end{align*}
\end{proof}

\section{Linear Forward Processes and input-affine reverse diffusions}\label{sneg}
In this section we consider the LTV forward model \eqref{vecforward} with $u_k\sim N(0,I)$. 
The state covariance, $\Sigma_k$ obeys
\begin{equation}\label{eq:TVcovariance}
	\Sigma_{k+1}=F_k\Sigma_kF_k^{\top}+G_kG_k^{\top}\quad k\geq k_0.
\end{equation}
We assume $\Sigma_{k_0}$ and hence all $\Sigma_{k}$ are positive definite. First for  the case where the initial and all subsequent states are Gaussian, viz $x_k\sim N(0,\Sigma_k)$,
we show that Theorem \ref{tJu} recovers the result of \cite{Kailath}. 
Observe
\[ 
U_k|X_{k+1}\sim N\left (G_k\Sigma_{k+1}^{-1}x_{k+1}, I-G_k^{\top}\Sigma_{k+1}^{-1}G_k)^{-1}\right ).
 \]
Note that the choice 
\[ 
\xi(v_{k+1})=\dfrac{\exp\left (-\frac{1}{2}v_{k+1}^{\top}(I-G_k^{\top}\Sigma_{k+1}^{-1}G_k)^{-1} v_{k+1}\right )}{(2\pi)^{n/2}\det\left ( I-G_k^{\top}\Sigma_{k+1}^{-1}G_k\right)}
\]
has a Fourier Transform. Now in \eqref{genv} choose
\begin{equation}\label{vlin}
	v_{k+1}=G_k^{\top}\Sigma^{-1}_{k+1}x_{k+1}-u_k. 
\end{equation}
Then \eqref{vlin} meets the conditions of Theorem \ref{tJu}, as in the notation of that theorem, $J_{hu}=-I$, $\Gamma=\mR^m$,  \eqref{gentrans2} holds and $\xi(\cdot)$ has a Fourier transform. Note this  $v_k$ is \textit{Gaussian}. Further, \eqref{vlin}  conforms to \eqref{vexpress} with $D_{k+1}(x_{k+1})=I$ and $C_{k+1}(x_{k+1})=G_k^{\top}\Sigma^{-1}_{k+1}x_{k+1}$. Thus, from Theorem \ref{taffine} the reverse model is necessarily input affine. In fact, using \eqref{eq:TVcovariance}, \eqref{eq:ABdef} and \eqref{Bx}, the reverse diffusion becomes
\begin{align*}
	x_k&=F_k^{-1}x_{k+1}-F_k^{-1}G_ku_k\\ &=\Sigma_kF_k^{\top}\Sigma_{k+1}^{-1}x_{k+1}+F_k^{-1}G_kG_k^{\top}\Sigma_{k+1}^{-1}x_{k+1}\\
	&-F_k^{-1}G_ku_k\\
	&=\Sigma_kF_k^{\top}\Sigma_{k+1}^{-1}x_{k+1}+F_k^{-1}G_kv_{k+1}
\end{align*}
with the last equality following from \eqref{vlin}.
This is indeed, the reverse model,  derived by very different means, in \cite{Kailath}. 

That $v_{k+1}=E[u_k|x_{k+1}]-u_k$ in (\ref{vlin}) and $x_{k+1}$ are independent, is unsurprising, as in this  Gaussian case  $v_{k+1}$ represents the innovations in estimating $u_k$ from $x_{k+1}$.  
This reverse model though is \textit{not unique}. Indeed, Section \ref{sex} provides a second reverse model in this LTV Gaussian case using Theorem \ref{tgeneralfamily}, which is neither input-affine nor has a Gaussian process noise $v_{k+1}$.

The rest of this section examines the existence of input-affine reverse models for LTV processes, \eqref{vecforward}, with $G_k$  square as in \eqref{linearnoising}, and shows that for a wide class of state densities such a reverse diffusion does not exist. We will associate $F_kx_k=z$, $x_{k+1}=x$, $u_k=u$, $D_{k+1}(\cdot)=-D(\cdot)$, $C_{k+1}(\cdot)=C(\cdot)$ and $v_{k+1}=v$. 

In view of Theorem \ref{taffine} an input-affine reverse diffusion exists iff there are
$C(\cdot)$ and $D(\cdot)$, the latter \textit{nonsingular} everywhere, such that
\begin{equation}\label{vcd}
	V=C(X)+D(X)U
\end{equation}
and $X$ are independent. As  $G_k$ is square,  from Assumption \ref{ass:two}, it is nonsingular. Then  there is no loss of generality in assuming
that
\begin{equation}\label{xzu}
	X=Z+U,
\end{equation}
 with $Z$ and $U\sim N(0,I)$ independent. By way of summary of the remainder of the section, following recording of key background results in Section \ref{sknown},   we show in Section \ref{sscalar} that in the scalar case  independence of $V$ and $X$ necessitates that $Z$ be Gaussian. Subsequently, for the vector case Section \ref{svec} shows that for a very wide class of non-Gaussian random vectors $Z$, $V$ and $X$ cannot be independent,  proving the nonexistence of input-affine reverse models for this very wide class. 
 
 \subsection{Some known results}\label{sknown}
 We first give a  lemma associated more commonly with the fluids literature, \cite{Zheligovsky_2014}. 
 
 \begin{lemma}\label{lanalytic}
 	Consider \eqref{xzu} with $Z\in\mR^n$ having a pdf $p_Z(z)$ and independent with $U\sim N(0,I)$. Then the pdf $p_X(x)$  of $X$ is a real analytic function.
 \end{lemma} 
 \begin{proof}
 	Since $p_U(u)$ the pdf of $U$ is Gaussian its complex extension is analytic. It thus satisfies the Cauchy-Riemann equations (CRE): i.e. for all $s\in \mC^n$ the partials of real and imaginary parts of $p_U(s)$ with respect to real and imaginary parts of entries of $s$ obey certain equalities. As $p_Z(z)$ is real for $z\in \mR^n$, these equalities must also hold for the partials of real and imaginary parts of $p_U(s-t)p_Z(t)$ with respect to real and imaginary parts of entries of $s$, for any $t\in \mR^n$.
    As $p_Z(\cdot)\in L^1$, the complex extension of $p_X(x)$,
 	\[
 	p_X(s)=\int_{\mR^n}p_Z(t)p_U(s-t)dt
 	\]
 	is finite for all $s$ in every compact subset of $\mC^n.$ Thus the integral commutes with its derivatives with respect to the entries of $s$. Thus $p_X(s)$ also satisfies the CREs, is analytic and its restriction to $\mR^n$ is therefore real analytic.
  \end{proof}
 
 The next Lemma is the consequence of the Tweedie-Robbins-Esposito identity, \cite{Esposito1968}.
 
 \begin{lemma}\label{lespositop}
 Under the condition of Lemma \ref{lanalytic} suppose $p_X(x)>0$ for all $x\in \mR^n$. Then 
 \[ 
 \mbox{Cov}(U|X=x)=I+\nabla^2\log p_X(x),
  \]
  where $\nabla^2\log p_X(x)$ is the Jacobian of $\nabla \log p_X(x)$.
 \end{lemma}
 \begin{proof}
From Lemma \ref{lanalytic} and the fact that $p_X(x)>0$, $ \log p_X(x)$ exists and is smooth. The Tweedie-Robbins-Esposito identity, \cite{Esposito1968} states that 
\[
E[U|X=x]=x+\nabla \log p_X(x).
\]
The result follows from \cite{Hatsell1971}, that shows \footnote{A reference yielding both formulas in the scalar case is \cite{Tweedie}.}:
 \[ 
 \mbox{Cov}(U|X=x)= \nabla_x E(U|X=x).
  \]
 \end{proof}
 The last result in this subsection links the Gaussian nature of $Z$ to Cov$(U|X=x)$ being constant on a nonempty open set.
 \begin{lemma}\label{ltweedie}
	Under the conditions of Lemma \ref{lespositop}, if Cov$(U|X=x)$  is constant on a nonempty open set then $Z$ is Gaussian.
\end{lemma}
\begin{proof}
	From Lemma \ref{lanalytic}, $p_X(x)$ and hence $\nabla^2_x\log p_X(x) $ are real analytic. Thus from Lemma \ref{lespositop},  Cov$(U|X=x)$ is also real analytic.  Hence, if it is  constant on a  nonempty open set,  it is so everywhere, as also is
	$\nabla^2_x\log p_X(x) $. In turn,  $\log p_X(x) $ is quadratic and so $X$ is Gaussian. Cramér, \cite{cramer1936}, states if the sum of two independent random variables is Gaussian then so are the summands.   Thus, as $Z$ and $U$ are independent $Z$ must be Gaussian.
\end{proof}
In view of this last result,  the rest of this section is devoted to proving that when $V$ in \eqref{vcd} and $X$ are independent then for a wide variety of cases Cov$(U|X=x)$ is constant on a nonempty open set. The scalar case is much more straightforward, and has a more complete result. Therefore we treat it first and separately.

 \subsection{The scalar case}\label{sscalar}
 In this section we consider the scalar case of \eqref{vcd} and \eqref{xzu} and make the following assumption. 
 
 \begin{assumption}\label{ass:one}
 		In \eqref{xzu}, $z\in \mR$ has a thrice differentiable pdf that is positive everywhere and 
 	 $u\sim N(0,1)$ and $z$ are independent.
 	
 \end{assumption}
 The main result of this subsection (which is intuitively reasonable but lacks at this point a short proof)  is as follows.
 
 \begin{theorem}\label{tscalargauss}
 	Suppose Assumption \ref{ass:one} holds and there exist thrice differentiable functions $c(\cdot):\mR\rightarrow \mR$ and   $ d(\cdot):\mR\rightarrow \mR$, with $d(x)$ nonzero everywhere such that with $x=z+u$, the random variables
 	\begin{equation}\label{v}
 		v=c(x)+d(x)u
 	\end{equation} 
 	and $x$
 	are independent. Then $z$ is Gaussian.
 \end{theorem}

 We remark that in the above theorem hypothesis, there is no restriction on the sign of $d(x)$. The theorem would however continue to be true if there is such a sign restriction. If it has been proved for the case $d(x)>0$, in a situation where $d(x)<0$, the independence of $v$ and $x$ would immediately imply independence of $-v$ and $x$, and then the theorem would apply.

 With the aid of several lemmas, we will first show that either $z$ is Gaussian or  $d(x)$ is a constant. Then we will argue that a constant $d(x)$ also implies $z$ is Gaussian.
 We have the following lemma.

 \begin{lemma}\label{lbasic}
 	Suppose the conditions of Theorem \ref{tscalargauss} hold. Then with $g(x)=1/d(x)$, 
 \begin{align}
 	&p_Z\big(x-g(x)(v-c(x))\big)p_U\big(g(x)(v-c(x))\big)|g(x)\nonumber\\
 	&=p_X(x)p_V(v) \label{pxv}
 \end{align}
 \end{lemma}

 \begin{proof}
 	
 	As $u$ is Gaussian, from Lemma \ref{lanalytic}, $p_Z(z)$ is real analytic. The mapping we are working with is:
 	\begin{align}
 		x&=z+u\\\nonumber
 		v&=c(x)+d(x)u=c(z+u)+d(z+u)u
 	\end{align}
 	for which the inverse mapping is
 	\begin{align}
 		u&=g(x)[v-c(x)]\\\nonumber
 		z&=x-g(x)[v-c(x)]=x+g(x)c(x)-g(x)v
 	\end{align}
 	Note,
 	\begin{align*}
\frac{\partial v}{\partial u}-\frac{\partial v}{\partial z}&=c'(z+u)+d'(z+u)u+d(z+u)\\
&-c'(z+u)-d'(z+u)u\\
&=d(z+u).
 	\end{align*}
 	Thus  
 	\begin{align*}
 		\det J(z,u)&=\begin{vmatrix} \frac{\partial x}{\partial z}&\frac{\partial x}{\partial u}\\
 			\frac{\partial v}{\partial z}&\frac{\partial v}{\partial u}
 		\end{vmatrix}=\begin{vmatrix}
 			1&1\\
 			\frac{\partial v}{\partial z}&\frac{\partial v}{\partial u}
 		\end{vmatrix}\\
 		&=\frac{\partial v}{\partial u}-\frac{\partial v}{\partial z}=d(z+u)
 	\end{align*}
 	Then it is standard that
 	\begin{equation}
 		p_{XV}(x,v)=p_{ZU}(z,u)/|J(z,u)|
 	\end{equation}
 	with 
 	\begin{align*}
 	&	p_{ZU}(z,u)=p_Z(z)p_U(u)\\
 		&=p_Z(x+g(x)c(x)-g(x)v)p_U\big(g(x)(v-c(x))\big)
 	\end{align*}
 	Hence
 	\begin{align*}
 		p_{XV}(x,v)&=p_Z\big(x+g(x)c(x)-g(x)v\big)\\
 		&\times p_U\big(g(x)(v-c(x))\big)|g(x)|
 	\end{align*}
 	Since $x$ and $v$ are independent by the lemma hypothesis, the lemma is proved. 
 \end{proof}

 We will now obtain a functional equation for the second derivative of the score function associated with $z$. Take two fixed but arbitrary distinct values of $x$, viz $x_1$ and $x_2$, but retain $v$ as a variable in \eqref{pxv}. Make the definitions
 \begin{equation}\label{simp}
 	g(x_i)=a_i,~\mbox{ and }x_i+g(x_i)c(x_i)=b_i.
 \end{equation}
 Taking the logarithm yields the two equations for $i=1,2$:
 \begin{align}
 &	\ln[p_Z(b_i-a_iv)]+\ln p_U\left(a_iv-a_ic(x_i)\right)+\ln |a_i|\notag\\
 &=\ln p_X(x_i)+\ln p_V(v)
 \end{align}
 Now subtract one equation from the other, and use the fact that $u$ is Gaussian. Certain quantities are of course constant while $\ln p_U\left(a_iv-a_ic(x_i)\right)$ is quadratic in $v$. There results
 \begin{equation*}
 	\ln p_Z(b_1-a_1v)-\ln p_Z(b_2-a_2v)=\pi(v)
 \end{equation*}
 where, crucially, $\pi(v)$ is  quadratic in $v$. Define
 \begin{equation}
 	w(\cdot ) =[\ln p_Z]^{(3)}(\cdot).
 \end{equation}
 Then,
 \begin{equation}\label{third}
 	a_1^3w(b_1-a_1v)=a_2^3w(b_2-a_2v), ~\forall ~v\in \mR.
 \end{equation}

 We now provide the following lemma, which effectively solves this functional equation.
 
 \begin{lemma}\label{laffine}
 	Suppose $\ell:\mR\rightarrow\mR$  is continuous and  there exist real $a\neq 0$ and $b$ such that for all $s\in \mR$
 	\[ 
 	\ell(s)=a^3 \ell(as+b).
 	\]
 	Then either $| a|=1$  or $\ell(s)\equiv 0$.
 \end{lemma}
 \begin{proof}
 	Let $s\in\mathbb R$ be arbitrary but fixed. For integer $n>0$ there holds:
 	\begin{equation}\label{recurse}
 		\ell(s)=a^{3n} \ell(a^ns+b\sum_{i=0}^{n-1}a^i)
 	\end{equation}
 	If $|a|< 1$ then 
 	for some $K$ 
 	\[
 	\lim_{n\rightarrow \infty}a^ns+b\sum_{i=0}^{n-1}a^i=K
 	\]
 	and 
 	using the continuity of $\ell$ we can conclude that the right side of \eqref{recurse} goes to zero as $n$ tends to infinity. Since $s$ is arbitrary, this implies $\ell(s)\equiv 0.$  If $| a|>1$ then  with $z=as+b$,
 	\[
 	\ell(z)=\frac{1}{a^3}\ell\left(\frac{z-b}{a}\right )
 	\]
 	and the same argument applies. Hence again, $\ell(s)\equiv 0.$
 \end{proof}
 
  We need one last lemma.
 \begin{lemma}\label{ldconst}
 	Under the conditions of Lemma \ref{lbasic}, either $d(x)$ is a constant or $z$ is Gaussian.
 \end{lemma}
 \begin{proof}
 	From Lemma \ref{lbasic} and the arguments after it \eqref{third} holds. Then by choosing  $s=-a_1v+b_1$,
 	\[ 
 	w(s)=a^3w(as+b), ~\forall s\in\mR
 	\]
 	for some  $b$ and $a=a_2/a_1.$  Note that because $d(x)$ is continuous and nonzero for all $x$, $a_2$ and $a_1$ have the same sign, and $a>0$.  If for any choice of $x_1$ and $x_2$, it holds that $a\neq 1$ then
 	from Lemma \ref{laffine} the third derivative of  $\ln p_Z(z)$ is a constant. Thus $\ln p_Z(z)$   is quadratic and $p_Z(z)$ is Gaussian. 
 	
 	Otherwise, $a=1$ implying $a_1=a_2$ or by \eqref{simp} $g(x_1)=g(x_2)$ or  $d(x_1)=d(x_2)$. This argument applies for all choices of $x_1=x_2$.   i.e.  $d(x)$ is a constant.
 \end{proof}
 We now provide the proof of Theorem \ref{tscalargauss}.
 
 \noindent
 {\bf Proof of Theorem \ref{tscalargauss}:}
 The independence of $x=z+u$ and $v=c(x)+d(x)u$  ensures that if $z$ is non-Gaussian, then $d(x)\equiv d.$  Then Lemma \ref{ltweedie}  yields   
 \[ 
 \mbox{Var}[u|x] =1+[\ln(p_X)]''(x)
 \]
 Thus
 \[ 
 \mbox{Var}[v|x]=d^2(1+[\ln(p_X)]''(x))
 \]
 
 Using again the independence of $v$ and $x$, it follows that Var$(V|X=x)$ and hence $[\ln(p_X)]''(x)$ 
 is a constant.   Thus from Lemma \ref{ltweedie}, $z$ is Gaussian.
 
 \subsection{The vector case with square $G_k$}\label{svec}
 The previous section we effectively proved that for the scalar case of \eqref{vecforward}, an input-affine reverse process exists iff the initial state is Gaussian. We now turn to the vector case with $G_k$ square as in \eqref{linearnoising} 
 with the following assumption.
 
 \begin{assumption}\label{akey}
 	The random vectors 
 	$Z\in \mR^n$ and $U\in \mR^n$ are independent, with $U\sim N(0,I)$ and $Z$ possessing a pdf. There holds $X=Z+U$. Moreover, in \eqref{vcd} for 
 	$C:\mR^n\rightarrow \mR^n$ and $D:\mR^n\rightarrow \mR^{n\times n}$, there exist $M_i$,  such that for all $x\in \mR^n$   $\|D(x)\|\leq M_2^{-1/2}$,
 	 $\|D^{-1}(x)\|\leq M_1$, Cov$(V)=I$; and $V$ and $X$ are independent. 
 \end{assumption}
 The assumptions on $D(\cdot)$ ensure that  $D(\cdot)$ is not just nonsingular but well-conditioned. As in Theorem \ref{taffine}  the reverse diffusion uses $D^{-1}(\cdot)$,  they thus ensure  the numerical stability of the input-affine reverse diffusion.
 
 The assumption of unity covariance of $V$ is without loss of generality as long as $\Sigma_V$, the covariance of $V$, is positive definite. This is so, as  $ V$, $C(x)$ and $D(x)$ can be replaced by $\Sigma_V^{-1/2}V$, $\Sigma_V^{-1/2}C(x)$ and $\Sigma_V^{-1/2}D(x)$ respectively, while maintaining  independence between $\Sigma_V^{-1/2}V$ and $X.$
 The first lemma we need is as follows. 
 
 \begin{lemma}\label{lpos}
 	Under Assumption \ref{akey}, $\Sigma_U(x):=\mbox{Cov}(U|X=x)$ is bounded and positive definite, obeying the equation
 	\begin{equation}\label{sqD}
 		\Sigma_U(x)=D^{-1}(x)D^{-\top}(x)
 	\end{equation}. 
 \end{lemma}
 \begin{proof}
 	From \eqref{vcd} 
 	$\mbox{Cov}(V|X)=D(x)\Sigma_U(x)D^{\top}(x)$. The result holds as   $\mbox{Cov}(V|X)=\mbox{Cov}\;V=I$. 
 \end{proof}

 We next provide an intermediate lemma. It introduces a conditional cumulant generating function, from which subsequently covariance information can be obtained. 
 
 \begin{lemma}\label{lint}
 	Adopt Assumption \ref{akey}. Then the conditional cumulant generating function (CCGF)
 	\begin{equation}\label{ccgf}
 		K_{U|X=x}(t,x)=\log E\left [\left . e^{t^{\top}U}\right | X=x\right ],
 	\end{equation} 
 	obeys for all $x\in \mR^n$ and $t\in \mR^n$, 
 	\[ 
 	K_{U|X=x}(t,x)=\log p_X(x-t)-\log p_X(x)+\frac12\|t\|^2,
 	\]
 	where $p_X(\cdot)$ is the pdf of $X$.
 \end{lemma}
 \begin{proof}
 	Note that with $X=Z+U$ where $Z$ and $U$  are independent, there holds
 	$p_{X,U}(x,u)=p_{Z}(x-u)p_U(u)$,
 	implying that 
 	\[
 	p_{U|X}(u|x)=\frac{p_Z(x-u)p_U(u)}{p_X(x)}
 	\]
 	From the Gaussian nature of $U=X-Z$  
 	\begin{align*}
 	&E\left [e^{t^{\top }U}|X=x\right ]=\int_{\mR^n}e^{t^{\top}u}\frac{p_Z(x-u)p_U(u)}{p_X(x)}du\\
 		&=\frac{1}{p_X(x)\sqrt{(2\pi)^n}}\int_{\mR^n}e^{t^{\top}(x-z)}p_Z(z)e^{-\frac12 \|x-z\|^2 }dz\\
 		&=\frac{e^{\frac12 \|t\|^2}}{p_X(x)\sqrt{(2\pi)^n}} \int_{\mR^n}e^{-\frac12 \left (\|t\|^2-2t^{\top}(x-z)+ \|x-z\|^2\right ) }p_Z(z)dz\\
 		&=\frac{e^{\frac12 \|t\|^2}}{p_X(x)\sqrt{(2\pi)^n}} \int_{\mR^n}e^{-\frac12 \|x-z-t\|^2 }p_Z(z)dz\\
 		&=\frac{e^{\frac12 \|t\|^2}p_X(x-t)}{p_X(x)}.
 	\end{align*}
 	The result follows by taking logarithms on both sides.
 \end{proof}

 We now establish a key recursion.
 \begin{lemma}\label{lkey}
 	Adopt Assumption \ref{akey}. Define the  cumulant generating function (CGF) of $V$:
 	\begin{equation}\label{cgf}
 		K_V(s)=\log E\left [e^{s^{\top}V}\right ].
 	\end{equation}
 	Then with $\Sigma_U(x)$ as defined in \eqref{sqD}
 	and for all $x\in \mR^n$ and $t\in \mR^n$, there holds
 	\begin{equation}\label{recurs}
 		\Sigma_U(x-t)=D^{-1}(x)\left(\nabla^2_t K_V(D^{-\top}(x)t) \right) D^{-\top}(x).
 	\end{equation}
 \end{lemma}
 \begin{proof}
 	
 	Using the independence of $V$ and $X$ we obtain a different expression for the conditional cumulative generating function introduced above: 
 	\begin{align*}
 		K_{U|X=x}(t,x)&=\log E\left [e^{t^{\top}D^{-1}(x)(-C(x)+V)}\right ]\\
 		&=-t^{\top}D^{-1}(x)C(x)+K_V(D^{-\top}(x)t).
 	\end{align*}
 	Now differentiate this equality twice with respect to $t$ (i.e. form the Hessian of each side), and then use Lemma \ref{lint} and Lemma \ref{lespositop}. We indeed obtain
 	\begin{align*}
 	&	D^{-1}(x)\left(\nabla^2_t K_V(D^{-\top}(x)t) \right) D^{-\top}(x)\\
 	&=\nabla^2_t K_{U|X=x}(x,t)\\
 		&= \nabla^2_t \log p_X(x-t)+I\\
 		&=\nabla^2_x \log p_X(x-t)+I\\
 		&=\Sigma_U(x-t).
 	\end{align*}
 \end{proof}
 
 We now make a second assumption that involves a condition on the behavior of the conditional covariance matrix $\Sigma_U(x)$, roughly speaking requiring there to be a particular direction (defined by a vector $\theta$ in the assumption) such that $\Sigma_U(x)$ evaluated for values of $x$ along a ray (defined by varying $a$ in the assumption) in this direction but commencing from any point ($b$ in the assumption) in a nontrivial fixed ball is asymptotically constant. The assumption is motivated by showing in the Appendix that it is satisfied when the pdf of $Z$ is a finite Gaussian mixture. 
 
 \begin{assumption}\label{aray}
 	There is an $M>0$,  a nonzero $\theta\in \mR^n $, and a $\Sigma_U^*(\theta)\in \mR^{n\times n}$ so that the following holds. For  every $\epsilon>0$, there is an $\eta(\epsilon,M)$, such that for all  $a>\eta(\epsilon,M)$ and all $b\in \mR^n$, with $\|b\|\leq M$,
 	\begin{equation}\label{rayconv}
 		\|\Sigma_U^*(\theta)-	\Sigma_U(a\theta +b)\|\leq \epsilon.
 	\end{equation}	
 \end{assumption}
 Observe $\eta$ depends on $M$ rather than on the precise affine shift $b$ and the limit point is independent of $b$.
 We now establish a key consequence of \eqref{rayconv}.
 
 \begin{lemma}\label{lKvlim}
 	Suppose Assumption \ref{aray} holds and the conditions of Lemma \ref{lkey} apply. Then  there exists $R>
 	0$ such that for all $s\in \mR^n$ in the ball $\|s\|\leq R$
 	\[ 
 	\nabla_s^2K_V(s)=I.
 	\]
 \end{lemma}
 \begin{proof}
 	First observe that for any  $\Delta\in \mR^{n\times n}$, 
 	because of the upper bound on $\|D(x)\|$ in Assumption \ref{akey} one has
 	\begin{align*}
 		\| \Delta\|&= \| D(x)[D^{-1}(x)\Delta D(x)^{-\top}]D(x)^{\top}\|\\
 		&\leq \|D(x)\| \|D^{-1}(x)\Delta D(x)^{-\top}\| \|D(x)^{\top}\|\\&\leq M_2^{-1}\|D^{-1}(x)\Delta D(x)^{-\top}\|.
 	\end{align*}
 	Thus, 
 	\begin{equation}\label{lowerbound}
 		\|D^{-1}(x)\Delta D(x)^{-\top}\|\geq M_2\|\Delta\|.
 	\end{equation}
 	Let $M>0$ and $\theta\in\mR$ be as in Assumption \ref{aray}. 	For  every $\epsilon>0$, there is  $\eta(\epsilon,M)$, such that for all  $a\geq \eta(\epsilon,M)$ and $b\in\mathbb R^n$ with $\|b\|\leq M$,
 	\begin{align*}
 		\epsilon&>\|\Sigma_U^*(\theta)-	\Sigma_U(a\theta +b)\|.
 	\end{align*} 
 	In addition, using the choice $b=0$, there holds 
 	\[
 	\epsilon>\| \Sigma_U^*(\theta)-\Sigma_U(a\theta)\|
 	\]
 	Consequently, with the given $M>0$ and $\theta\in\mathbb R^n$, there holds for all $a\geq \eta(\epsilon, M)$ and $b$ with $\|b\|\leq M$, there holds the inequality
 	\[
 		2\epsilon	>\|\Sigma_U(a\theta)-	\Sigma_U(a\theta +b)\|
    \]
    Now use from Lemma \ref{lpos} the equation \eqref{sqD} and from Lemma \ref{lkey} the equation \eqref{recurs} to conclude
    \begin{align*}
 		2\epsilon&=\|D^{-1}(a\theta)\left [I-\nabla^2_bK_V(-D^{-\top}(a\theta)b)\right ]D^{-\top}(a\theta)\|\\
 		&\geq M_2\|I-\nabla^2_bK_V(-D^{-\top}(a\theta)b)\|,
 	\end{align*}
 	where we have used \eqref{lowerbound} for the inequality.
 	 Also from  Assumption \ref{akey}, it is clear that for some suitably small ball $\|s\|\leq R$, the set of associated $b=D^{\top}(a\theta)s$ will lie in $\|b\|\leq M$. Thus, for every $\epsilon>0$ and all $s$ with $\|s\| \leq R$, 
 	\[ 
 	\|\nabla^2_s(K_V(s))-I\|\leq O(\epsilon).
 	\]
 	The lemma holds as  $\epsilon$ can be arbitrarily small.
 \end{proof}
 We  now give the main result of this subsection.
 \begin{theorem}\label{tconj}
 	Under Assumptions \ref{akey} and \ref{aray}, $Z$ is Gaussian.
 \end{theorem}
 \begin{proof}
 	From Lemma \ref{lKvlim}, \eqref{recurs} and \eqref{sqD}, for some $R>0$, some $x\in\mR^n$, and all $t\in \mR^n$ having the property that $\|D^{-1}(x)t\|\leq R$,
 	\begin{align*}
 		\Sigma_U(x-t)&=D^{-1}(x)\left(\nabla^2_t K_V(D^{-\top}(x)t) \right) D^{-\top}(x)\\
 		&=D^{-1}(x) D^{-\top}(x)=\Sigma_U(x).
 	\end{align*}
 	Thus, Cov$(U|X=x)$ is constant in a nonempty open set containing this $x\in \mR^n$.  Lemma \ref{ltweedie} yields the result.
 \end{proof}
 
 Thus the key conclusion is that under Assumption \ref{aray} a forward linear diffusion \eqref{vecforward} diffusion with square $G_k$ will only have an input-affine reverse model  when there is a Gaussian initial state.  We prove in the appendix that  Assumption \ref{aray} is indeed satisfied when $Z$ is a Gaussian Mixture. If an initial state in \eqref{vecforward} is a Gaussian Mixture so are all subsequent states. Thus for such states and square $G_k$ an input affine reverse model does not exist.
 

\section{A Gaussian Mixture Example}\label{sex}
An important class of state pdfs are Finite Gaussian Mixture Models (GMM), described by
\begin{equation}\label{GMM}
	\sum_{i=1}^N \alpha_i\, N(\mu_i,S_i),
	\qquad
	\alpha_i>0,\ \sum_{i=1}^N \alpha_i=1.
\end{equation} 
In \eqref{vecforward} for example, if $u_k\sim N(0,I)$ and the initial state is a GMM then so are all subsequent states. Specifically, if $x_k$ has the pdf in \eqref{GMM}, then 
\[ 
X_{k+1}\sim \sum_{i=1}^N \alpha_i\, N\left (F_k\mu_i,F_kS_iF_k^{\top}+G_kG_k^{\top}\right ).
 \]
 
 Such pdf modeling is important for  a nonlinear forward process as well. This is so as under reasonable assumptions most  pdfs can be approximated by GMMs.
The work of Parzen \cite{parzen1962estimation} pioneered the approximation of pdfs using convex combinations of pdfs of standard forms, including Gaussian pdfs. In much more detail, 
\cite{li1999mixture} for multivariate densities argued that approximation using a Gaussian density mixture was best measured using a Kullback-Leibler distance measure, and argued that the approximation error goes to zero at a rate $1/k$ where $k$ is the number of Gaussian densities making up the mixture. Recently, \cite{van2024mixtures} provided algorithms for finding the weights, means and variances of such a mixture minimizing the Kullback-Leibler distance to the prescribed multivariate density. (Of course  certain reasonable side conditions on the prescribed density must be satisfied).

Section \ref{sneg} shows that when $x_{k+1}$ is modeled as a GMM the reverse process cannot be input affine even for  \eqref{vecforward} at least when $G_k$ is square. As a consequence,   formulating the reverse model should involve conditional cdfs using Theorem \ref{tspec} or  Theorem \ref{tgeneralfamily}.

Such  conditional cdfs are easily  found convex combinations of Q-functions (complementary cdfs of Gaussians) with state dependent weights. To avoid   distractions from the main ideas, we illustrate with a  scalar example: $x_{k+1}=ax_k+u_k$
with $u_k\sim N(0,1)$ and
\begin{equation}\label{gm}
	X_k\sim 0.5\left (N\left (0,\sigma^2/a^2\right )+N \left (\mu/a,\sigma^2/a^2\right )\right)
\end{equation}
with $\mu\neq 0$. Call $x=x_{k+1}$ and $u=u_k$. 
Define:
\begin{equation} \label{alpha}
	0\leq\alpha(x)=\dfrac{e^{-\frac{x^2}{2(\sigma^2+1)}}}{e^{-\frac{x^2}{2(\sigma^2+1)}}+e^{-\frac{(x-\mu)^2}{2(\sigma^2+1)}}}\leq 1.
\end{equation}
Then  with $s=\sigma/\sqrt{\sigma^2+1}$,  $p_{U|X}(u|x)$ equals
\begin{equation}\label{pux}
	\dfrac{1}{\sqrt{2\pi}~s}\left (\gamma(x)e^{-\frac{\left (u-\frac{x}{\sigma^2+1}\right )^2}{2s^2}}+(1-\gamma(x))e^{-\frac{\left (u-\frac{x-\mu}{\sigma^2+1}\right )^2}{2s^2}}\right )
\end{equation}
Theorem \ref{tspec}  shows that
\[ 
v_{k+1}=P_{U|X}(u_k|x_{k+1})
\]
allows definition of a  reverse diffusion. We can define this function more explicitly. 
Recall the \textit{Q-function}, \cite{papoulis1965probability}
\[ 
Q(u)=1-\dfrac{1}{\sqrt{2\pi} }\int_{-\infty}^ue^{-\frac{z^2}{2}}dz \mbox{ and }
Q(-u)=1-Q(u).
\]
Then using \eqref{pux} and the definition of $s$ above it $v_{k+1}=P_{U|X}(u_k|x_{k+1})$ is given by
\begin{align*}
	v_{k+1}&=\gamma(x_{k+1})Q\left (\frac{x_{k+1}-(\sigma^2+1)u_k}{\sigma\sqrt{\sigma^2+1}}\right )\\
	&+(1-\gamma(x_{k+1}))Q\left (\frac{x_{k+1}-(\sigma^2+1)u_k-\mu}{\sigma\sqrt{\sigma^2+1}}\right )
\end{align*}
and regarding this as an implicit equation defining $u_k=\psi_{k+1}(x_{k+1},v_{k+1})$ in terms of $v_{k+1}$ and $x_{k+1}$, we conclude for the reverse model
\begin{equation*}
	x_k=a^{-1}x_{k+1}-a^{-1}\psi_{k+1}(x_{k+1},v_{k+1})
\end{equation*}
The expression for $\psi_{k+1}$ must be numerically determined from the values of $x_{k+1}$ and $v_{k+1}$.

Consider $\mu=0,$  i.e., when $p_{X_k}(x)$ collapses to a Gaussian. In this case $v_{k+1}=Q((x_{k+1}-(\sigma^2+1)u_k)/(\sigma\sqrt{\sigma^2+1}))$, which  is very different from  the solution given in Section \ref{sneg}, underscoring the fact that reverse diffusions are not unique.

\section{Conclusion}\label{sconc}
Generative AI is hampered by the absence of a theory for directly reversing a stochastic difference equation. We show that a  nonlinear forward diffusion, with possibly non-Gaussian nonzero mean initial density, will have a reverse diffusion if in (\ref{genv}) every $x_{k+1}$, $v_{k+1}$ pair uniquely determines $u_k$, and $v_k$ and $x_k$ are independent, and that the reverse diffusion of an LTV forward process will be input-affine only if \eqref{genv} is affine in $v_{k+1}$.  

We focus on the technical challenge of meeting the independence requirement. For that we provide a necessary and sufficient condition and two means of meeting it using  cdfs of entries of $u_k$ conditioned on $x_{k+1}$ and some other entries of $u_k$. We show that scalar linear forward processes cannot have an input-affine reverse diffusion unless the initial state is Gaussian. When the input matrix $G_k$ in \eqref{vecforward} is square, we also show that an input-affine reverse process does not exist for a wide class of state densities including finite Gaussian mixtures.

We are pursuing several areas in current and future work that have been opened up. First, we conjecture that an LTV forward process has an input-affine reverse diffusion if and only if the initial state is Gaussian. This paper goes a long way toward proving that fact, but gaps remain. It would also be important to consider input-affine nonlinear forward processes, viz. the discrete time counterparts of  SDEs like \eqref{eq:forward}, considered in \cite{anderson1982reverse}. When the state sequence is Gaussian in \eqref{vecforward}, an analytically tractable reverse process is available. Otherwise, the ones we have found require numerical solutions, and lack a closed form. Whether, there are non-Gaussian initial densities for which such a simple forward diffusion as  \eqref{linearnoising}, gives rise to a closed form solution for the reverse diffusion is worth exploring though not critical. We also wish to consider ideas which are known to be valid in the continuous-time case \cite{anderson1982reverse}. When is a reverse-time model of a reverse-time model the original forward model? If a forward model is stationary on an infinite-time interval, does the reverse-time model have the same property? And if a forward model forgets initial probabilities, does a reverse-time model exhibit similar forgetting in the reverse direction. Another issue  is the connection between the continuous-time and discrete-time constructions when a continuous-time model is time-discretized with a very short discretization interval. 

Perhaps the most important future direction to be explored is to make this theory directly applicable to Generative AI, which lacks knowledge of the cdfs  involving $u_k$ conditioned on $x_{k+1}$. Rather,  Generative AI algorithms use the score, $\nabla_x \log p_{X_k}(x)$ of the state sequence, which is  obtained using neural networks, \cite{NEURIPS2020_DDPM} from empirical data. One approach could be to focus on \eqref{linearnoising} and formulate  neural networks that find  score functions of entries of $u_k$ conditioned on $x_{k+1}$ and other of its entries. Such a network can perhaps be used to directly find $\psi_{k+1}(\cdot,\cdot)$ in \eqref{psi} using Theorem \ref{tspec}. When the state sequence is Gaussian, papers like \cite{ZAMAN2026109380} effect approximate reversal in one step instead of a thousand steps. Given that we have a theory for exact reversal, it is perhaps possible to obtain such a one step reversal through our approach even in the non-Gaussian case. In all this we should exploit the fact that \eqref{linearnoising}, widely used in the Generative AI literature, is really a set of scalar equations carrying obvious simplifications even though the entries of the initial state are dependent.

\ifCLASSOPTIONcaptionsoff
  \newpage
\fi



%


\bibliographystyle{IEEEtran}
\bibliography{discretereversetime}

\appendix
We now show that $Z$ is a finite Gaussian Mixture then Assumption \ref{aray} holds.
We begin with  Lemma \ref{lSdiff} that will help establish the asymptotic dominance of a solitary term in $\Sigma(x)$, along a set of affine rays. The pairs $(\mu_k,Q_k)$ used in the Lemma are actually the mean and inverse covariance of distinct Gaussian densities arising later in this section. 

\begin{lemma}\label{lSdiff}
	Consider for $k\in \{1,\cdots, N\}$,  a set of distinct pairs $(\mu_k,Q_k)$ where $\mu_k\in \mR^n$ and  $Q_k\in \mR^{n\times n}$ is positive definite symmetric.
	Then  there exist a set $\Omega\subset \{1,\cdots, N\}$, an index $i\in \Omega$ and a vector    $y\in \mR^n$ such that, first, for all $j\in \Omega$, $Q_j=Q$; second,  
	\begin{equation}\label{nonstrict}
		y^{\top}Qy < y^{\top}Q_ky ,~\forall ~  k\notin \Omega
	\end{equation}
	and third,  either $\Omega=\{i\}$ or
	\begin{equation}\label{mustrict}
		y^{\top}Q\mu_i> y^{\top}Q\mu_j, ~\forall j\in \Omega\setminus \{i\}.
	\end{equation}
\end{lemma}
\begin{proof}
	We consider two cases, the first where the $Q_k$ are all distinct, and the second where some (or even all) may be equal. If the $Q_k$ are distinct then the set of $y\in \mR^n$ where for $i\neq j$ the pairwise equalities  $y^{T}Q_iy=y^{T}Q_jy$ hold  are hypersurfaces of dimension at most $n-1$ each. As there are a finite number of such surfaces, there must be a $y$ for which all $y^{T}Q_ky$  are distinct. Then for some $i$ 
	\[ 
	y^{T}Q_iy<y^{T}Q_jy, ~\forall~j\neq i.
	\]
	The result thus holds when $\Omega=\{i\}$.
	
	Suppose then the $Q_i$ are not distinct. Then there is a partition  $\Omega_1, \cdots, \Omega_L$ of $\{1,\cdots, N\}$ such that all the $Q_k$ indexed by elements of a given   $\Omega_l$ are equal, i.e. 
	\[ 
	Q_k=Q_j, ~\forall \{k,j\}\subset \Omega_l.
	\] Arguing as above there must be one subset of the partition, $\Omega=\Omega_J$ say, and a $y\in\mathbb R^n$ for which \eqref{nonstrict} holds. Recalling that $Q_j=Q$ say, for all  $j\in \Omega$ and  that the pairs $(\mu_j, Q)$ are distinct, so then are the $\mu_j$ for $j\in \Omega.$ Now the set of $y$ satisfying \eqref{nonstrict} is in fact an open set and because $Q$ is positive definite symmetric,   the corresponding set of $Qy$ is also open. Consequently. as the  $\mu_j$ are distinct there must be a $Qy$ in this open set for which $y^{\top}Q\mu_j$ for $j\in \Omega$ are distinct. Then \eqref{mustrict} holds.
	
\end{proof}

We now take a further step towards establishing the dominance  motivating  Lemma \ref{lSdiff}, by introducing the framework allowing consideration of a weighted finite sum of distinct Gaussian densities. 
\begin{lemma}\label{lweights}
	Adopt the hypothesis of Lemma \ref{lSdiff}. With $\alpha_i>0$ and $\sum_{i=1}^{m}\alpha_i=1$,
	define 
	\[
	\gamma_i(x)
	=
	\alpha_i\,\phi\left (x;\mu_i,Q_i^{-1}\right )/\sum_{k=1}^N \alpha_k\,\phi\left (x;\mu_k,Q_k^{-1}\right ).
	\]
	where $\phi\left (x;\mu_i,Q_i^{-1}\right )$ is the probability density function
	\[ 
	\phi\left (x;\mu_i,Q_i^{-1}\right )=e^{-\frac{1}{2}(x-\mu_i)^{\top}Q_i(x-\mu_i)}/\sqrt{\det(2\pi Q_i^{-1})}.
	\]
	Then there exists a single $i\in \{1,\cdots,N\}$,  and a vector $y\in \mR^n$, with the following property:  for some  $M>0$ , there exist $\delta(M)>0$, $\nu(M)>0$, and $r(M)\in \mR$ such that for all $b\in\mR$ with $\|b\|\leq M,$ and $a>\delta(M)$,  
	\begin{equation}\label{gammabound}
		1-	\gamma_i(ay+b)=\sum_{k\neq i}\gamma_k(ay+b)\leq  (N-1)e^{-\nu(M)a-r(M)}.
	\end{equation}
	
\end{lemma}
\begin{proof}
	The equality part of \eqref{gammabound} follows (actually for any $y,a$ and $b$) from the fact that the $\gamma_i(x)>0$ sum to one.
	 Define $\rho_i(x)=\log \gamma_i(x)$ and set \(ay+b=x\). For each \(i\), 
	\begin{align*}
		\rho_i(ay+b)
		&= \log \alpha_i
		-\dfrac{1}{2}\log\left [\det(2\pi Q_i^{-1})\right ]\\
		&
		-\frac12(ay+b-\mu_i)^T Q_i (ay+b-\mu_i)\\
		&-\log\sum_{k=1}^N \alpha_k\phi(ay+b;\mu_k,Q_k^{-1}).
	\end{align*}
	The last term, being common to all $i$, cancels in the expression for any differences of two $\rho_i(x)$ and will henceforth be  able to be ignored. The third term, a quadratic form in $y$, becomes $\beta_i(a)$ given by
    \begin{small}
	\[
	\beta_i(a)=-\frac12 a^2\, y^T Q_i y
	-a\, y^T Q_i (b-\mu_i)
	-\frac12 (b-\mu_i)^T Q_i (b-\mu_i).
	\]
    \end{small}
	Since the pairs $ (\mu_i, Q_i)$ are distinct and the $Q_i$ are positive definite, the conditions of Lemma \ref{lSdiff} hold. The Lemma implies there  exists a set $\Omega$,  a particular index $i\in\Omega$ and a $y\in\mathbb R^n$ for which the inequalities of the Lemma hold. Use these values and through renumbering if necessary, suppose the particular index is 1. It is now necessary to consider two cases, one where there exists $j\notin \Omega$, and one for which there exists no such $j$, i.e. all $Q_j$ are the same. 
	
	In the first case,  for $j\notin \Omega$,  \eqref{nonstrict} ensures that
	there are $\kappa_j>0$, $\theta_j(b)$, affine in $b$ and $\chi_j(b)$, quadratic in $b$ such that for all $a>0$
	\[
	\beta_1(a)-\beta_j(a)\geq \kappa_j a^2+ \theta_j(b) a +\chi_j(b).
	\]
	Thus for all $j\notin\Omega$,
	\begin{align*}
		&\rho_1(ay+b)-\rho_j(ay+b)\geq \kappa_j a^2+ \theta_j(b) a +\chi_j(b)+\log \alpha_1\\
		&-\dfrac{1}{2}\log\left [\det(2\pi Q_1^{-1})\right ]-\log \alpha_j+\dfrac{1}{2}\log\left [\det(2\pi Q_j^{-1})\right ].
	\end{align*}

	Note one can find a sufficiently large $a$, such that for all $\|b\|\leq M$ and $j\notin\Omega$,  $\rho_j a+ \theta_j(b)>0$. More precisely, one can find a  $\delta_1(M)>0$ and a $\nu_1(M)>0$  such that for all $a>\delta_1(M)$ and all $\|b\|\leq M$
	\begin{align*}
		\rho_j a^2+ \theta_j(b) a =a(\rho_j a+ \theta_j(b))\geq a\nu_1(M).
	\end{align*}
	In other words, 
	with  $\|b\|\leq M$, there are   $\delta_1(M)>0$, $r_1(M)$ and  $\nu_1(M)>0$ such that for all $a>\delta_1(M) $ 
	\[
	\rho_1(ay+b)\geq \rho_j(ay+b)+ a\nu_1(M)+r_1(M),~\forall j\notin\Omega.
	\]
	The second case to consider arises when $j\in\Omega$, implying  $Q_j=Q$. Forming $\beta_1(a)-\beta_j(a)$ as before, the quadratic term in $a$ no longer appears, and the affine term in $a$ has sign determined by the inequality of \eqref{mustrict}. Moreover, this affine term is no longer dependent on $b$. In particular, for all $j\in \Omega\setminus \{1\}$, 
	there are $\rho_j>0$, and $\chi_j(b)$ quadratic in $b$ such that for $a>0$
	\[ 
	\beta_1(a)-\beta_j(a)\geq \rho_j a +\chi_j(b).
	\] 
	
	In this case with  $\|b\|\leq M$, there are   $\delta_2(M)>0$, $r_2(M)$ and  $\nu_2(M)>0$ such that for all $a>\delta_2(M) $
	\[ \rho_1(ay+b)\geq \rho_j(ay+b)+ a\nu_2(M)+r_2(M),~\forall j\\\in\Omega\setminus \{1\} \]
	It is then straightforward to combine the two cases to conclude that there exist $\delta(M), r(M)>0$ and $\nu(M)$ such that for all $a>\delta(M) $
	\begin{equation}\label{relgam2}
		\rho_1(ay+b)\geq \rho_j(ay+b)+ a\nu(M)+r(M),~\forall j\neq 1.
	\end{equation}
	This inequality is equivalent to
	\[ 
	\gamma_1(ay+b)\geq \gamma_j(ay+b)e^{\nu(M)a+r(M)},~\forall j\neq 1,
	\]
	or equivalently
	\[ 
	\gamma_j(ay+b)\leq \gamma_1(ay+b)e^{-\nu(M)a-r(M)},~\forall j\neq 1,
	\]
	from which through summation we obtain
	\begin{align*}
		1-	\gamma_1(ay+b)&=\sum_{j=2}^N\gamma_j(ay+b)\\
		&\leq (N-1)\gamma_1(ay+b)e^{-\nu(M)a-r(M)}\\
		&\leq (N-1)e^{-\nu(M)a-r(M)},
	\end{align*}
	as $\gamma_1(ay+b)\leq 1.$ The result   follows with $i=1$.
\end{proof}

We need one last Lemma that provides an expression for Cov$(U|X=x)$ when $Z$ is a finite Gaussian mixture.
A very similar expression is used in \cite{bar-shalom}, though we could not find a derivation there.

\begin{lemma} \label{lcondcov}
	Under Assumption \ref{akey} suppose that the random variable $Z$ has a density $\sum_{i=1}^N\alpha_iN(\mu_i,S_i)$ with $\alpha_i>0,\sum_i\alpha_i=1$. as in \eqref{GMM},
	with $\mu_i\in \mR^n$, positive definite  $S_i=S_i^{\top}\in \mR^{n\times n}$.  Consider   $\gamma_i(x)$ and $\phi(x;\mu_i,Q_i)$ as in Lemma \ref{lweights} with $Q_i=(S_i+I)^{-1}$. Define
	\begin{equation}\label{mui}
		m_i(x)=(S_i+I)^{-1}(x-\mu_i),~\Delta_i=S_i(S_i+I)^{-1},
	\end{equation}
  and
    \begin{equation}\label{mbar}
    \bar m(x)=\sum_{k=1}^{N}\gamma_k(x)(m_k(x).
	\end{equation} 
	Then $\Sigma_U(x)=Cov(U|X=x)$ is given by
   
	\begin{align}
		\Sigma_U(x)&=\sum_{k=1}^{N}\gamma_k(x)\Delta_k\notag\\
		&+\sum_{k=1}^{N}\gamma_k(x)\big(m_k(x)-\bar m(x)\big)\big(m_k^{\top}(x)-\bar{m}^{\top}(x)\big)\label{Sigma}
	\end{align}
\end{lemma}
\begin{proof}

Introduce a latent variable $\cI$ corresponding to each of the $N$ components of the mixture. 	Conditioned on $\cI=k$, the joint gaussianity of $X,Z$ and $U$ yields
\begin{align*}
\mbox{Cov}(U|x,\cI=k)&=\Delta_k\\
E[U|x,\cI=k]&=Q_k(x-\mu_k)=m_k(x)
\end{align*}
and the posterior probability of component $k$ is 
\[
\gamma_k(x)=P(\cI=k|X=x)=\frac{\alpha_kN(\mu_k,Q_k^{-1})}{\sum_{k=1}^N\alpha_kN(\mu_k,Q_k^{-1})}
\]
By the law of total covariance,
\begin{align*}
    \Sigma_U(x)&=E[\mbox{Cov}(U|(x,\cI)]+\mbox{Cov}(E[U|x,\cI])]\notag\\
    &=\sum_{k=1}^N\gamma_k\Delta_k+\sum_{k=1}^N[m_k(x)-\bar m(x)][m_k^{\top}(x)-\bar m(x)^{\top}]
\end{align*}

\end{proof}
We now show that if in Assumption \ref{akey}, $Z$ and hence $X$ is a Gaussian Mixture then Assumption \ref{aray} holds.

\begin{theorem}\label{tray}
	Under Assumption \ref{akey} suppose the pdf of  $Z$ is as in Lemma \ref{lcondcov}.
	Then Assumption \ref{aray} holds with $\Sigma_U(x)=\mbox{Cov}(U|X=x)$ as defined in Lemma \ref{lpos}.
\end{theorem}

\begin{proof}
	Adopt the notation of Lemma \ref{lweights} with $\gamma_i(x)$ as given there. Identifying $Q_k^{-1}=(S_k+I)$,  by Lemma \ref{lweights} there is an $i\in \{1,\cdots, N\}$,   $y\in\mR^n$ and $M>0$, for which the following holds:
	With $a\in \mR$,    there are $ \delta(M)>0$, $r(M)$ and $\nu(M)>0$, such that for all $a\geq\delta(M)$, and $b\in \mR^n$ with $\|b\|\leq M$, \eqref{gammabound} holds.
	
	Suppose $a, y,$ and $b$ are such that \eqref{gammabound} holds. Relabel  $i$ as 1. Consider the quantities defined in Lemma \ref{lcondcov}. Use Lemma \ref{lweights} and the fact that $m_k(x)$ has an affine dependence on $x$.
	For $k\neq 1$, $x=ay+b$ as in \eqref{gammabound}, $\|b\|\leq M$ one has from  \eqref{mui} and \eqref{mbar} the following: For all $a>\delta(M)$ and some $\nu(M)>0$, $r(M)$, $K_i(M)>0$,
	\begin{align*}
		&\left \|\sum_{k=1}^{m}\gamma_k(x)m_k(x)m_k^{\top}(x)-\bar{m}(x)\bar{m}^{\top}(x)\right \|\\
		&\leq\gamma_1(x)(1-\gamma_1(x))\left \|m_k(x)m_k^{\top}(x)\right \|\\
		&
		+\left \|\sum_{k=2}^{m}\gamma_k(x)m_k(x)m_k^{\top}(x)\right \|\\
		&+\left  \|\sum_{k=2}^{m}\sum_{j=1}^{m}\gamma_k(x)\gamma_j(x)m_k(x)m^{\top}_j(x)\right \|\\
		&\leq 	 K_1(M)e^{-a\nu(M)}(K_2(M)a^2+K_3(M)a+K_4(M)).
	\end{align*}
	Thus from \eqref{Sigma}  with $y$ and $b$ as in Lemma \ref{lweights}, for all $a>\delta(M)$ and some $\nu(M)>0$, $K_i(M)>0$
	\begin{align*}
	&	\|\Sigma_U(ay+b)	-\Delta_1\|\\
	&\leq (1-\gamma_1(ay+b))\|\Delta_1\|+\left\|\sum_{k=2}^{m}\gamma_k(ay+b)\Delta_k\right\|\\
		&+K_5(M)e^{-a\nu(M)}(K_6(M)a^2+K_7(M)a+K_8(M))\\
		&\leq K_{9}(M)e^{-a\nu(M)}(K_{10}(M)a^2+K_{11}(M)a+K_{12}(M)).
	\end{align*}
	Thus Assumption \ref{aray} holds, with $\Sigma_U^*(y)=\Delta_1$.
\end{proof}

\end{document}